\documentclass{AIAA} % {document}
\usepackage{latexsym}
\usepackage{amsmath}
\usepackage{amssymb}
\usepackage{stmaryrd}
\usepackage{graphicx}
\usepackage{mathrsfs}
\newcommand{\h}{{\mathbf{h}}}
\newcommand{\x}{{\mathbf{x}}}
\newcommand{\y}{{\mathbf{y}}}
\newcommand{\vv}{{\mathbf{v}}}
\newcommand{\uu}{{\mathbf{u}}}
\newcommand{\f}{{\mathbf{f}}}
\newcommand{\p}{{\mathbf{p}}}
\newcommand{\s}{{\mathbf{s}}}
\newcommand{\M}{{\mathcal{M}}}
\newcommand{\N}{{\mathcal{N}}}
\newcommand{\LL}{{\mathcal{L}}}
\newcommand{\CR}{{\mathcal{CR}}}
\newcommand{\RR}{{\mathcal{R}}}
\newcommand{\diag}{{\text{diag}}}
\newcommand{\eig}{{\text{eig}}}
\newcommand{\g}{\mathcal{S}}

\graphicspath{{./}}
\usepackage{color}
\usepackage{amsmath}
\usepackage{amssymb}
\usepackage{graphicx}
\usepackage{comment,xspace}
\usepackage{hyperref}
\usepackage{fancybox}
\usepackage{setspace}
\usepackage{multirow}

\newcommand{\real}{{\mathbb{R}}}
\newcommand{\reals}{\real}

\newtheorem{theorem}{Theorem}[section]
\newtheorem{proposition}[theorem]{Proposition}
\newtheorem{lemma}[theorem]{Lemma}
\newtheorem{remark}[theorem]{Remark}
\newtheorem{example}[theorem]{Example}
\newtheorem{definition}[theorem]{Definition}
\newtheorem{corollary}[theorem]{Corollary}

\newcommand{\mymargin}[1]{{\color{red}*} \marginpar{\color{red}\small\ttfamily $\Lsh$ #1}}
\renewcommand{\mymargin}[1]{}

\begin{document}

\title{Generalized cyclic algorithms for formation acquisition and control}

\author{Jaime L. Ramirez
\footnote{Department of Aeronautics and Astronautics, Massachusetts Institute of Technology} and Jean-Jacques Slotine
\footnote{Nonlinear Systems Laboratory, Massachusetts Institute of Technology}}
\affiliation{
Massachusetts Institute of Technology, 77 Massachusetts Avenue, Cambridge, MA 02139, U.S.A}

\begin{abstract}
  This paper presents a new approach to distributed nonlinear control
  for formation acquisition and maintenance, inspired by recent
  results on cyclic topologies and based on tools from contraction
  theory. First, simple nonlinear control laws are derived to achieve
  global exponential convergence to basic symmetric formations. Next,
  convergence to more complex structures is obtained using control
  laws based on the idea of convergence primitives, linear combinations of basic control elements. All control laws use only local information and communication to achieve a desired
  global behavior.
\end{abstract}

\maketitle 

\section{Introduction}
Multi-vehicle systems which achieve mission objectives by cooperatively
controlling their relative positions have been widely studied in the
distributed control literature \cite{Scharf2004}. Many applications
require achieving a formation. Some of them can as well achieve
the overall mission objective by converging to some manifold without a
need to specifically track relative trajectories. These include, for example,
fragmented telescopes, packs of chaser robots and space missions where
only a structural shape of the formation is important like the Laser
Interferomatric Space Antenna (LISA) mission.
 
There are specific advantages in achieving a formation by using only
{\it local} information, rather than a central coordinating
mechanism. Especially, as the number of agents in the formation is increased, 
combining the information of the whole formation to calculate the control commands becomes challenging.
Thus, recent research has studied the emergence of a
global behavior based on local rules \cite{Jadbabaie2003,
Olfati-saber2006}. However, research on formation control has
mostly focused on approaches that converge to fixed point trajectories by
tracking fixed relative positions to neighboring agents.

The most common approach to formation control studied in the
literature defines laws based on tracking relative position to a set
of neighbors, using a variety of methods to synthesize the controller, e.g.  LQG, H$_\infty$, etc. This approach is not always the most desirable, and the
control effort can often be significantly reduced by eliminating the
'unnecessary' constraints in the formation degrees of freedom, and
converging instead to desired manifolds defined by linear or nonlinear
constraints. Additionally, if the emergent behavior is an overall
geometric state with some unconstrained degrees of freedom, a leader
or a pair of leaders can control those states for the whole formation
without the need of a global coordination mechanism reassigning
relative position targets.

Some authors have studied the convergence to relevant symmetric formations by
using potential functions, e.g. \cite{Olfati-saber2002, Sepulchre2008}. A main
pitfall in that case is convergence to local equilibria, leading to a lack of global
convergence guarantees and unpredictability of the behavior under
disturbances, with these difficulties exacerbated in the time-varying
case.  Guaranteed global convergence to formation, in itself a highly
desirable property, also has important implications in the robustness
of the formation architecture.

In this paper a new approach to distributed formation
acquisition and maintainance with global guarantees is discussed.  By contrast to
more explicit control techniques, this approach does not specify fixed
points trajectories for the formation. Rather, it allows the agents to
achieve a formation under specified {\it structural} constraints and
provides an approach to deriving sufficient conditions for global
convergence. The approach produces global convergence results for
nonlinear control laws and combinations of cyclic grids that achieve
complex formations while maintaining the global convergence
guarantees.

The control algorithms are inspired as a generalization of the cyclic
pursuit control laws recently presented by Pavone and Frazolli
\cite{Pavone2007} and extended by Ren \cite{Ren2008a} which converge to rotating formations. The analysis, based on contraction theory~\cite{Lohmiller1998,Wang2005},
yields global convergence results and direct extensions for nonlinear
systems and more general dynamic cases. It also allows the
introduction of {\it convergence primitives}, where control laws
consist of combinations of simpler control laws, and constraints in
the desired configuration are described as convergence subspaces and
combined to achieve more complex formations.

Since the main objective of the article is the introduction to the system analysis
approach, the focus is on simple integrator dynamics, but also illustrate implementation of the control algorithms for higher order dynamic systems.

Section~\ref{sec:Mathback} defines the notation and some mathematical
tools and background of the theoretical
approach. Section~\ref{sec:IntroToApproach} introduces the contraction
theory approach to the problem and its relations to the previous
literature. Section~\ref{sec:CyclicCont} derives control laws for
global convergence to regular polygonal formations and presents a
result for global convergence to a regular formation of specific size.
Section~\ref{sec:FormPrim} introduces convergence results based on
control primitives, which are presented through applications in
Section~\ref{sec:applications}. Brief concluding remarks are presented
in Section~\ref{sec:Conclusion}.

%%%%%%%%%%%%%%%%%%%%%%%%%%%%%%%%%%%%%%%%%%%%%%%%%%%%%%%%%%%%%%%%%%%%%%%%%%%%%
\section{Background}\label{sec:Mathback}
%%%%%%%%%%%%%%%%%%%%%%%%%%%%%%%%%%%%%%%%%%%%%%%%%%%%%%%%%%%%%%%%%%%%%%%%%%%%%
In this section, definitions and results from matrix theory that used in the derivation of the article's results are provided.

\subsection{Notation}
Let $\reals_{>0}$ and $\reals_{\geq 0}$ denote the positive and nonnegative real numbers, respectively, and denote as $\real(x)$ the real part(s) of the complex element $x$. Let $j \doteq \sqrt{-1}$ and $I_{n}$ denote the identity matrix of size $n$; $A^\top$ and $A^*$ denote, respectively, the transpose and the conjugate transpose of a matrix $A$. A block diagonal matrix with block diagonal entries $A_i$ is denoted $\diag[A_i]$. For an $n \times n$ matrix $A$, let $\text{eig}(A)$ denote the set of eigenvalues of $A$, and we refer to its $k$th eigenvalue as $\lambda_k(A)$.  $\lambda_{max}(A+A^\top) \geq \lambda_{k}(A+A^\top)$ for all $k$ and correspondingly $\lambda_{min}(A): \lambda_{min}(A+A^\top) \leq \lambda_{k}(A+A^\top)$.  In general, a matrix $A(\x,t)$ is positive definite if $\lambda_{min}(A(\x,t)) > 0$, and it is denoted as $A > 0$, equivalently negative definiteness an positive semidefiniteness are defined and respectively denoted $A < 0$, $A \geq 0$.
The state of a single agent $i$ is denoted by $\x_i$ which in general $\x_i \in R^3$ and the overall state of the system will be denoted as $\x=[\x_1^{T},\x_2^{T},\ldots,\x_n^{T}]^{T}$. 

\begin{definition}{Flow-invariant manifolds}\\
A flow invariant manifold $\M$ of a system $\dot \x = f(\x)$ is a manifold such that if $\x(t_o) \in \M$ then $\x(t>t_o) \in \M$. We are interested in flow invariant manifolds that can described as the nullspace of a smooth operator, $\bar \x \in \M \Rightarrow V(\bar \x,t) = 0, \frac{d}{dt}(V(\bar \x,t)) = 0$.
\end{definition}

\subsection{Kronecker Product }\label{sec:kron}
Let $A$ and $B$ be $m \times n$ and $p \times q$ matrices, respectively. Then, the Kronecker product $A \otimes B$ of $A$ and $B$ is the $mp\times nq$  matrix
\begin{equation*}\label{eq:krno}
A \otimes B = \begin{pmatrix}
a_{11}B \, & \,\, \ldots \, \,& \, a_{1n}B\\
\vdots & & \vdots\\
a_{m1}B & \ldots & a_{mn}B
\end{pmatrix}.
\end{equation*}
If $\lambda_A$ is an eigenvalue of $A$ with associated eigenvector $\nu_A$  and $\lambda_B$ is an eigenvector of $B$ with associated eigenvector $\nu_B$, then $\lambda_A\lambda_B$ is an eigenvalue of $A \otimes B$ with associated eigenvector $\nu_A \otimes \nu_B$. Moreover: $(A \otimes B)(C \otimes D) = AC \otimes BD$, where $A$, $B$, $C$ and $D$ are matrices with appropriate dimensions.

\subsection{Rotation Matrices}
A rotation matrix is a real square matrix whose transpose is equal to its inverse and whose determinant is +1. The eigenvalues of a rotation matrix in two dimensions are $e^{\pm j\alpha}$, where $\alpha $ is the  magnitude of the rotation. The eigenvalues of a rotation matrix in three dimensions are $1$ and $e^{\pm j\alpha}$, where $\alpha $ is the  magnitude of the rotation about the rotation axis; for a rotation about the axis $(0, 0 ,1)^T$, the corresponding eigenvectors are $(0, 0, 1)^T, (1, +j, 0)^T (1, -j, 0)^T$. We denote $R(\alpha)$ or $R_{\alpha}$ a rotation matrix of angle $\alpha$.

\subsection{Circulant Matrices}
A circulant matrix $C$ is an $n\times n$ matrix having the form
\begin{eqnarray}
C = \begin{pmatrix}
c_0    & c_1 & c_2 & \ldots &c_{n-1}\\
c_{n-1}& c_0 & c_1 & \ldots &\vdots \\
\vdots &      &      &          & c_1 \\
c_1    & c_2& \ldots &\dots &c_0
\end{pmatrix}
\end{eqnarray}
The elements of each row of $C$ are identical to those of the previous row, but are shifted one position to the right and wrapped around. The following theorem summarizes some of the properties of circulant matrices.
\begin{theorem}[Adapted from Theorem 7 in \cite{Matrix}] \label{thrm:circProp}
Every $n\times n$ circulant matrix $C$ has eigenvectors
\begin{eqnarray}\label{circEig}
\psi_k = \frac{1}{\sqrt{n}}\left( 1, e^{2\pi j k /n},\ldots, e^{2\pi j k(n-1)/ n} \right)^T, \quad k\in\{0,1,\ldots,n-1\}, \label{eq:eigvecCirc}
\end{eqnarray}
and corresponding eigenvalues
\begin{equation}\label{eq:eigCircMat}
\lambda_k = \sum_{p=0}^{n-1}c_p e^{2\pi j k p/n},
\end{equation}
and can be expressed in the form $C = U\Lambda U^*$, where $U$ is a unitary matrix whose $k$-th column is the eigenvector $\psi_k$, and $\Lambda$ is the diagonal matrix whose diagonal elements are the corresponding eigenvalues. Moreover, let $C$ and $B$ be $n \times n$ circulant matrices with eigenvalues  $\{\lambda_{B,k}\}_{k=1}^n$ and $\{\lambda_{C,k}\}_{k=1}^n$, respectively; then,
\begin{enumerate}
\item $C$ and $B$ commute, that is, $CB = BC$, and $CB$ is also a circulant matrix with eigenvalues $\text{\emph{eig}}(CB)=$ $\{\lambda_{C,k} \,\lambda_{B,k}\}_{k=1}^n$;
\item $C + B$ is a circulant matrix with eigenvalues $\text{\emph{eig}}(C+B)= \{\lambda_{C,k} +\lambda_{B,k}\}_{k=1}^n$.
\end{enumerate}
\end{theorem}
From Theorem \ref{thrm:circProp} all circulant matrices share the same eigenvectors, and the same matrix $U$ diagonalizes  \emph{all} circulant matrices.

\subsection{Block Rotational-Circulant Matrices}\label{sec:rotcirc}
The set of matrices that can be written as $L\otimes R$ where $L$ is a circulant matrix and $R$ is a rotation matrix about a fixed axis all belong to a group of matrices that we denote as $\CR$. The set $\CR$ forms a conmutative matrix algebra, since the eigenvectors of $L$ are the same for all circulant $L$ following Section \ref{sec:kron} and the eigenvectors of $R$ are the same for all rotation matrix $R$ that share an axis of rotation.

\subsection{Contraction theory} \label{subsec:ContTheo} The basic
contraction theory analysis tool is the result derived in
\cite{Lohmiller1998}, which we state here in a simplified form.
\begin{theorem}{\bf Contraction \cite{Lohmiller1998}}\\
Consider the deterministic system in $\real^n$:
\begin{eqnarray}
\dot \x = \f(\x, t) \label{eq:basicCsys}
\end{eqnarray}
where $\mathbf f$ is a smooth nonlinear function. Denote the Jacobian
matrix of $\mathbf f$ with respect to $\x$ by $\frac{\partial \mathbf f
}{\partial \x}$. If there exists an invertible square matrix $\Theta(\x,t)$, 
such that the matrix:
 \begin{eqnarray}
\mathbf F = \Theta \frac{\partial \mathbf f }{\partial \x} \Theta^{-1}
\end{eqnarray}
is uniformly negative definite, then all the system's trajectories converge exponentially to a single trajectory. The system is said to be contracting.
\end{theorem}
An important corollary of this approach is the definition of auxiliary
systems which have as particular solutions the trajectories of the
system being analyzed. Proving that the auxiliary (or virtual) system is
contracting shows that the system of interest converges
to trajectories of the auxiliary system. This is the principle of {\it partial contraction analysis} \cite{Wang2005}, which in particular leads to the following theorem.
\begin{theorem}{\bf Partial contraction \cite{Pham2007}}\label{thm:PartialContraction}\\
Consider a flow-invariant subspace $\M$, and the associated orthonormal projection matrix $V$ onto the orthogonal subspace $\M^\perp$. All trajectories of system (\ref{eq:basicCsys}) converge exponentially to $\M$ if the system 
\begin{eqnarray}
\dot \y =  Vf(V^T\y,t)
\end{eqnarray}
is contracting.
\end{theorem}
In this paper,  as in \cite{Pham2007}, the flow-invariant linear subspace $\M$ will typically
represent some synchronized behavior. The result also extends to certain types of nonlinear manifolds.

%%%%%%%%%%%%%%%%%%%%%%%%%%%%%%%%%%%%%%%%%%%%%%%%%%%%%%%%%%
\section{Contraction theory approach to formation control}\label{sec:IntroToApproach}
%%%%%%%%%%%%%%%%%%%%%%%%%%%%%%%%%%%%%%%%%%%%%%%%%%%%%%%%%%

First, we frame the  decentralized convergence to formation and point out some of the differences in the structure of the approach to the more general approach.

Consider a basic distributed control law that converge to a formation of vehicles to a regular formation equally phased from each other. Using the most common neighbor differences approach, which in a generalized way is \cite{Olfati-saber}:
\begin{eqnarray}
	\dot \x_i = \sum_{j \in \N_i} w_{ij}(\x_{j} - \x_{i} - \h_{ij}(t))
	\label{eq:ConsensusApproach}
\end{eqnarray}

The agents converge to fixed point trajectories of states defined by the time functions $\x_j(t)-\x_i(t)=\h_{ij}(t)$ which are to be agreed upon. The approach presented by Chung et al. \cite{Chung2009} uses contraction theory to show the convergence to a synchronized formation. In that approach, the vehicles synchronize their trajectories by tracking phase-separated trajectories and in essence their approach generalize a consensus approach for Lagrangian systems. 

Another approach consists on defining the desired geometry of the formation as the stationary points of gradient based laws of potential function \cite{Olfati-saber2002, Paley2008, Izzo2007}. This control approach however, lacks global convergence guaranties. If the information graph is not rigid, such approaches can converge to several different equilibria.

In a different manner, control laws based on the cyclic pursuit approach \cite{Pavone2007}:
\begin{eqnarray}\label{eq:basicCP}
\dot x_i &=& =  kR(\alpha)(\x_{i+1} - \x_i)
\end{eqnarray}
globally converge to circular formations and do not require to specify (and agreee upon) fixed trajectories of the states. This allows for degrees of freedom in the formation to be unconstrained, which leads to a reduction in the control effort. The convergence properties of the cyclic pursuit law have been presented in previous work  and several extensions have also been discussed \cite{Ramirez2009}. 

In a more general sense, in this paper we propose an approach to analyse the convergence properties of systems of the type:
\begin{eqnarray*}
	 \dot \x_i &=& f(\x)+ \uu_i\\
	 \uu_i &=& \sum_i A_{ij}(\x,t)( \x_j - \x_i)
\end{eqnarray*}
by defining invariant manifolds to which convergence is to be shown and using the results from partial contraction theory to show the convergence properties.

As an introduction to the contraction theory approach, consider the relation to a case presented in the context of contraction theory in the work of Pham and Slotine \cite{Pham2007}. In the context of contraction theory they presented an example which shows some relation to the cyclic pursuit algorithm. Specifically, the example in sec. 5.3 in \cite{Pham2007} analyzes a system of three coupled nonlinear oscillators:
\begin{eqnarray}
	\dot \x_1 &=& f(\x_1) + k(R(2\pi/3)\x_2 - \x_1)) \nonumber\\
	\dot \x_2 &=& f(\x_2) + k(R(2\pi/3)\x_3 - \x_2)) \label{eq:Pham53} \\
	\dot \x_3 &=& f(\x_3) + k(R(2\pi/3)\x_1 - \x_3)) \nonumber
\end{eqnarray}
where $R(\alpha)$ is a rotation matrix for an angle $\alpha$ and studies its convergence to the flow-invariant manifold:
\begin{eqnarray}\label{M_Pham}
\M &=& \left\{ \bar \x = \begin{bmatrix}R(\frac{2\pi}{3}) \x^\top, & R(\frac{\pi}{3})\x^\top&  \x^\top \end{bmatrix}^\top, \x \in \reals^2 \right\}
\end{eqnarray}

Now, considering a first order system with a basic cyclic pursuit control law in eq. (\ref{eq:basicCP}):
\begin{eqnarray}
\dot \x_i &=& = f(\x_i) - kR(\alpha)(\x_{i+1} - \x_i)\nonumber \\
\Rightarrow  {\dot \x} &=& \diag[f(\x_i)] - (kL_1\otimes R(\alpha)) \x =\diag[f(\x_i)] - \LL_3\x \label{eq:3CP}
\end{eqnarray}
where $L_1$ is the cyclic Laplacian:
\begin{eqnarray}\label{eq:L1}
L_1 = \begin{pmatrix} 1 &-1& 0&\ldots& 0\\ 0& 1 & -1& \ldots& 0\\ \vdots & & &&\vdots\\ -1 &0& 0 & \ldots& 1  \end{pmatrix}.
\end{eqnarray}
and $\M$ is an invariant set of $\diag[f(x_i)]$, i.e. $\diag[f(\bar x_i)] = 0$ for all $\bar \x \in \M$.

Analyzing the case of 3 nodes as well, the most obvious approach would be to consider the same invariant manifold in eq.~\ref{M_Pham}. The calculation of the contaction properties to such manifold show an inconclusive result though, as the corresponding Jacobian of the dynamics of $\y = V \x$ is only negative semi-definite.
Instead, consider a different flow-invariant subspace, namely, the set of regular n-polygon formations with origin anywhere in the plane:
\begin{eqnarray} \label{eq:flowInvariantM}
\M_n = \{ \bar \x: (\x_{i+1} - \x_i) = R_{2\frac{\pi}{n}}(\x_{i+2} - \x_{i+1})  \}, \quad \forall i < n-1
\end{eqnarray}
where $R_{\alpha} = R(\alpha)$. $\M_n$ describe a manifold of all the states where the vehicles are in a regular polygon formation as shown in figure \ref{fig:M1}.
\begin{figure}
\centering
  \includegraphics[width=5cm]{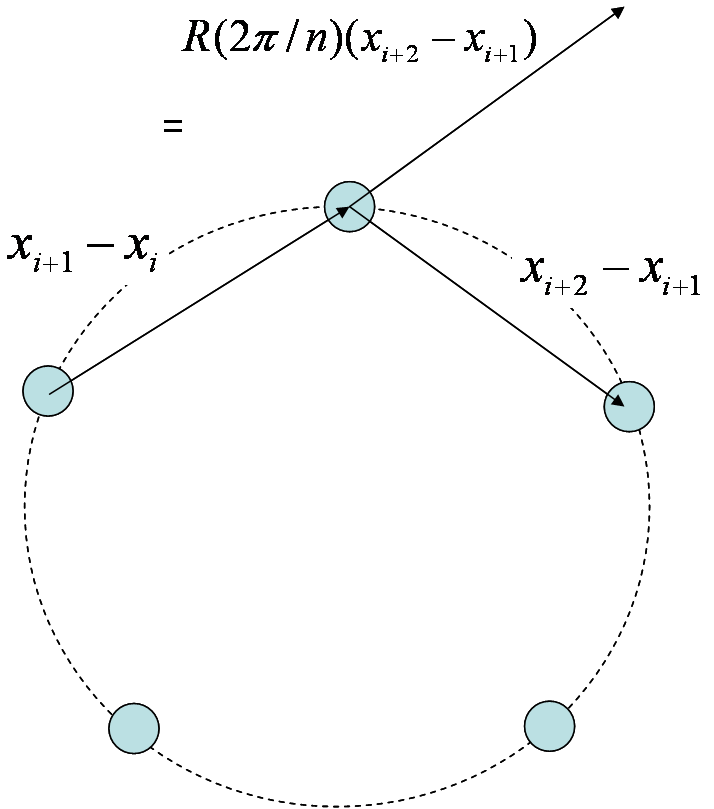}
  \caption{Constraint description of $\M_1$}\label{fig:M1}
\end{figure}
Notice that $\M \subset \M_n$.

$\M_n$ can be shown to be a flow-invariant manifold for the cyclic pursuit law:
\begin{eqnarray}
\dot \x_{i+1} - \dot \x_i &=& R(\alpha)(\x_{i+2} - \x_{i+1}) -  R_{\frac{2\pi}{n}}(\x_{i+1} - \x_{i})\nonumber \\
&=& R_{\frac{2\pi}{n}}R(\alpha)(\x_{i+3} - \x_{i+2}) -  R_{\frac{2\pi}{n}}R(\alpha)(\x_{i+2} - \x_{i+1})\nonumber \\
&=& R_{\frac{2\pi}{n}}(\dot \x_{i+2}) -  R_{\frac{2\pi}{n}}(\dot \x_{i+1})\nonumber \\
\dot \x_{i+1} - \dot \x_i &=& R_{\frac{2\pi}{n}}(\dot \x_{i+2} - \dot \x_{i+1}) \label{eq:MisInvManifold}
\end{eqnarray}

A matrix $V_{rn}$, such that $V_{rn}\bar \x = 0 \Leftrightarrow \bar \x \in \M_n$ is:
\begin{eqnarray}\label{Vrn}
V_{rn} = \begin{pmatrix} -I & I+R_{\frac{2\pi}{n}} & -R_{\frac{2\pi}{n}} & 0 &\ldots& 0  \\  0 & -I & I+R_{\frac{2\pi}{n}} & -R_{\frac{2\pi}{n}} & 0 &\ldots  \\ \vdots & \vdots & \vdots & \vdots & \vdots & \vdots \end{pmatrix}
\end{eqnarray}

In the case of 3 vehicles, with $\alpha = \pi/3$, $V_{r3}$ corresponds to:
\begin{eqnarray}
V_{r3} &=& \begin{pmatrix} -1& 0& 3/2 &\sqrt{3}/2 & -1/2& -\sqrt{3}/2 \\ 0 & -1 & -\sqrt{3}/2 & 3/2 & \sqrt{3}/2 & -1/2 \end{pmatrix},
\end{eqnarray}
and the projected Jacobian is:
\begin{eqnarray}
V_{r3}\LL_{3}V_{r3}^\top =\begin{pmatrix} 6 &3\sqrt{3} \\ 3\sqrt{3}& 6 \end{pmatrix}>0
\end{eqnarray}
which from Theorem \ref{thm:PartialContraction} verifies the global convergence of system in eq. \ref{eq:3CP} to manifold $\M_n$ if $\lambda_{max} (V_{r3} \frac{d\f(\x)}{d\x} V_{r3}^\top ) < 6$. Similar results are also found for more than three vehicles and in the next section we show a generalized result for any number of vehicles and more complex cyclic interconnections under state and time varying coupling matrices $R(\x,t)$.

%%%%%%%%%%%%%%%%%%%%%%%%%%%%%%%%%%%%%%%%
\section{Cyclic controllers for convergence to formation} \label{sec:CyclicCont}
%%%%%%%%%%%%%%%%%%%%%%%%%%%%%%%%%%%%%%%%

%%%%%%%%%%%%%%%%%%%%%%%%%%%%%%%%%%%%%%
\subsection{ Generalized cyclic approach to formation control} \label{subsec:secondOrder}

Theoretical results for the convergence to symmetric formations based on control laws that generalize the cyclic pursuit algorithm to more general interconnections and nonlinear cases are presented in this section.  First we show the global convergence of a basic control law to regular polygons under a generalized nonlinear cyclic topology with any number of vehicles and then we show how this result directly verifies the global convergence to rotating circular formations in the case of the basic cyclic pursuit algorithm.

This control law generalizes the results and allows for the design of distributed algorithms that converge to formations with geometric characteristics that depend on a common coordination state, and can be time varying. One can think for example a satellite formation that expands, contracts (by varying $\alpha$) or speeds up (by varying $k$) as a function of its location in orbit.

Consider the first order system $\dot \x = f(\x) + \mathbf u$ and the  generalized symmetric cyclic control law:
\begin{eqnarray}\label{eq:controllaw}
\uu_i(\x,t) &=&  \sum_{m\in \N_r} k_m(\x,t)\left(R_m(\x,t)( \x_{[i+m]} - \x_i) + R_m^\top(\x,t)(\x_{[i-m]} - \x_i )\right)
\end{eqnarray}
where $\N_r$ is a set of relative neighbors in the ordered set $\{1,..,N\}$, and $[p] \in \{1,..,N\}$ indicates $p$ modulo $N$.  The expression in eq. (\ref{eq:controllaw}) indicated that for each link $\{i,[i+m]\}$ there is a symmetric link $\{i,[i-m]\}$. $k_m(\x,t) \in \reals_{\geq 0}$ is a gain and $R(\x,t)$ is a coupling matrix that can be selected to achieve different behaviors. A general description of the overall dynamics of a system can then be writen as:
\begin{eqnarray}\label{eq:OveralRegPoly}
\dot \x &=& \f(\x) - \sum_m k_m(\x,t)((L_m \otimes R_m(\x,t) + L_m^\top \otimes R_m^\top(\x,t)))\x \nonumber \\
&\doteq& \f(\x) -\sum_m k_m(\x,t) (\LL_m(\x,t)+\LL^\top_m(\x,t))\x
\end{eqnarray}
and let us define:
\begin{eqnarray}
 \LL_{sm}(\x,t) &\doteq& (\LL_m(\x,t)+\LL^\top_m(\x,t))
\end{eqnarray}
where $\x$ is the vector describing the overall state of the system and $L_m$ are m-circulant Laplacian matrices describing cyclic underlying topologies with interconnections to each $m$-other agent as shown in fig. \ref{fig:ExplLaplacians}.
\begin{figure}
	\centering
		\includegraphics[width=0.6\textwidth]{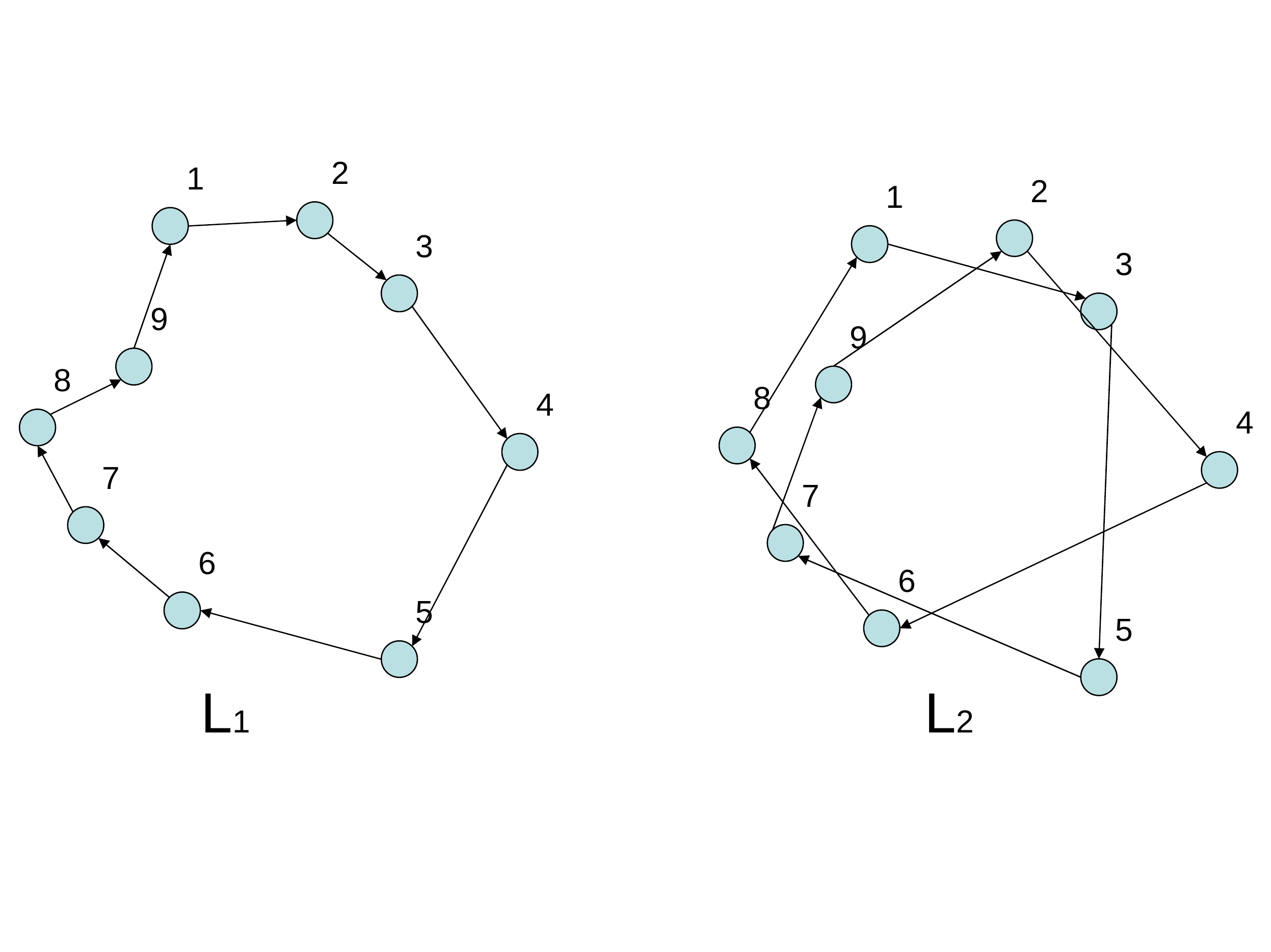}\\
	\caption{ \label{fig:ExplLaplacians} Different cyclic topologies} 
\end{figure}

Now, consider the manifold $\M_n$, presented in section \ref{sec:IntroToApproach}:
\begin{eqnarray} \label{eq:flowInvariantM}
\M_n = \{ \bar \x: (\x_{i+1} - \x_i) = R_{\frac{2\pi}{n}}(\x_{i+2} - \x_{i+1})  \} =  \{ \bar \x: V_{rn}\bar \x =0  \},
\end{eqnarray}
%with: 
%\begin{eqnarray}\label{eq:Vr}
%V_{rn} &=& \begin{bmatrix} -I & I+R_{\frac{2\pi}{n}} & -R_{\frac{2\pi}{n}} & 0 &\ldots& 0  \\  0 & -I & I+R_{2\frac{\pi}{n}} & -R_{2\frac{\pi}{n}} & 0 &\ldots  \\ \vdots & \vdots & \vdots & \vdots & \vdots & \vdots \end{bmatrix}
%%(1/2)( \left[ L_1 \otimes R(\alpha)\right] + \left[ L_1^\top  \otimes R^\top (\alpha)\right] )\nonumber \\ \cdot&&
%%\begin{bmatrix} -I & 0 &\ldots \\ I+R^\top (\alpha) & -I & \ldots \\ -R^\top (\alpha) & I+R^\top (\alpha) & \ldots \\ 0 & -R^\top (\alpha) & \ldots \\ \vdots & \vdots & \vdots  \\ 0 & 0 & .  \end{bmatrix}
%\end{eqnarray}
which can be shown (straightforward from eq. (\ref{eq:MisInvManifold})) to be a flow-invariant manifold of the dynamic system $\dot \x = k(\x,t) (\LL_{sm}(\x,t)) \x$.

Note that $V_{rn}$ can be written for the general case of $n$ vehicles in $3D$ as:
\begin{eqnarray}
V_{rn}&=& (W \otimes I_3)(I_n \otimes R_\eta)((L_1\otimes R_{\pi/N}) + (L^\top_1\otimes R^\top_{\pi/N})) \doteq  W_n \RR_\eta \LL^{(N)}
\end{eqnarray}
with $W = [I_{n-3} | 0_{3n-3}]$, $L_1$ is a cyclic Laplacian for a 1-circulant topology and $R_{\eta}$ is a rotation matrix for any value $\eta$.

Since $V_{rn}$ is full row rank, there exists an invertible transformation $V_{rn}=T\bar V_{rn}$ such that if $\bar V_{rn}A(\bar V_{rn})^\top>0$ then $V_{rn}AV_{rn}^\top>0$, where the columns of $\bar V_{rn}$ are a set of  orthonormal basis such that $\bar V_{rn}(\bar V_{rn})^\top = I$.
Then, following the results of partial contraction theory described in sec \ref{sec:IntroToApproach}, we have that if:
\begin{eqnarray}
 W_n \RR_{\eta} \LL^{(N)} \left(\frac{d\f(\x)}{dx} - \sum_m k_m(\x,t)\LL_{sm}(\x,t)\right) (\LL^{(N)})^\top \RR_{\eta}^\top W_n^\top < 0\quad uniformly
\end{eqnarray}
then, the system converges to $\M_n$. 

Since $\RR_{\eta}$, $\LL^{(N)}$, $\LL_{sm}(\x,t) \in \CR$, following the result in sec.\ref{sec:rotcirc} the calculation of the eigenvalues of their product and correspondingly verifying that $ W_n \RR_{\eta} \LL^{(N)} (\sum_m k_m(\x,t)\LL_{sm}(\x)) (\LL^{(N)})^\top \RR^\top_{\eta} W_n^\top < 0\ uniformly$ is straightforward and is shown in appendix \ref{app:A}. The results are applied in the following theorem:

\begin{theorem}{\bf Distributed nonlinear approach for global convergence to symmetric formations}\\
Consider the distributed system with a generalized cyclic topology, using control law in eq. (\ref{eq:controllaw}):
\begin{eqnarray}
\dot \x_i &=& \f(\x) + \mathbf u_i(\x,t) \\
\mathbf u_i(\x,t) &=& \sum_{m\in \N_r} k_m(\x,t)(R_(\alpha_m(\x,t))(\x_{[i+m]} - \x_i) + R(-\alpha_m(\x,t))(\x_{[i-m]} - \x_i))
\end{eqnarray}
for which a description of the overall dynamics is written as:
\begin{eqnarray}
\dot \x = \f(\x) - \sum_m k_m(\x,t) \LL_{sm}(\x,t) \x
\end{eqnarray}
$\x$ is the vector describing the overall state, $\bar \x$ is a particular regular polygonal state, and $V_{rn}\f(\bar \x)=0$.
Then, if:
\begin{equation}
\begin{split}\label{eq:condIBS}
&\sup_{\p,t} \left\{ \lambda_{max} \left( \RR_{\eta}\LL^{(N)}\frac{\partial \f(p)}{\partial x}(\LL^{(N)})^\top \RR^\top_\eta\right) \right.\\
&\left. - \min_{{\buildrel {1<i<N}\over{q\in\{-1,0,1\}}}}\!\!\!4\left(1\!-\!\cos\left(\frac{2i\pi}{N} \right)\right)\sum_{m}k_m(\p,t)\left(\cos\left(q\alpha_m(\p,t)\right)\! -\! \cos\left(q\alpha_m(\p,t)\!-\frac{2im\pi}{N}\right)\right)\right\}<0 
\end{split}
\end{equation}
for some $\eta$, the system globally converges to a regular polygon.

Specifically, if $\f(\x) = 0$ the system globally converges to a regular polygon formation if:
\begin{eqnarray}
\sup_{\p,t} \sum_{m} k_m(\p,t)(\cos(q\alpha_m(\p,t)) - \cos(q\alpha_m(\p,t) - 2im\pi/N)) &>& 0\quad
\end{eqnarray}
$\forall q \in \{-1,0,1\},\ i \in\{2,N-1\}$.
\end{theorem}

\begin{proof}
For a regular formation $V_{rn}\bar \x=0$, $V_{rn} \LL_{sm}(\bar \x,t)\bar \x=0$, then $V_{rn}( \f(\bar \x) - \sum k_m \LL_{sm}(\bar \x,t)\bar \x)= 0$, therefore $\M_n$ is an invariant manifold of the system.

From the results in Section \ref{prop:EigenLap} it is also true that $\lambda \big\{ \RR \LL \sum_m k_m(\x,t)\LL_{sm}(\x,t)  \LL^\top \RR^\top \big\} \geq \min_{{\buildrel {1<i<N}\over{q\in\{-1,0,1\}}}} \lambda_{ik}$ defined in eq. \ref{prop:EigenLap}. If $\RR \LL X \LL \RR< 0$, then $W_n \RR_\eta \LL X \LL \RR_\eta W_n^\top< 0$ for any matrix $X$.

Thus, if condition in eq. (\ref{eq:condIBS}) is satisfied, then  $V_{rn} (\frac{\partial \f(x)}{\partial \x} - \sum_m k_m(\x,t) \LL_{sm}(\x,t)_m) V_{rn}^\top < 0$ and therefore $\x$ exponentially converges to $\bar \x \in \M_n$.
\end{proof}

If $\f(\x)=0$, the above conditions in $\f(\bar \x)$ are satisfied and the global convergence to $V_{rn} \x=0$ is verified for specific combinations of $k_m(\x,t)$, $\alpha_m(\x,t)$. For the case of a 1-circulant topology $L_1$, if $k(\x,t)>0$ and $|\alpha(\x,t)| < 2\pi/N$, convergence to $\M_n$, i.e. a symmetric formation is guaranteed.

\begin{remark}
The Laplacian $\LL_{sm}(\x,t)$ is the symmetric part of the Laplacian $\LL_{m}(\x,t) = L_m\otimes R$. From the properties of positive negative matrices $VAV^\top < 0$ iff $V(A+A^\top)V^\top < 0$.

It was shown that the manifold $\M_{n}$ is an invariant manifold of the cyclic pursuit control law $R(\alpha)(\x_{[i+m]}-\x_i)$, then, verifying the conditions for convergence for the symmetric control law $\uu = \sum k_m (\x,t)(\LL_m(\x,t) + \LL_m(\x,t)^\top)\x$ is a direct proof of convergence to circular rotating formations for the directed topology $\uu = \sum k_m(\x,t) \LL_m(\x,t) \x $, with asymmetric control law $\uu_i = k_m(\x,t) R(\alpha)(\x_{[i+m]}-\x_i)$ generalizing the results for cyclic pursuit. 
\end{remark}
This last result verifies the convergence to rotating regular formations resulting for $k,\alpha$ constants, agreeing with the results obtained through linear analysis presented in \cite{Pavone2007, Ramirez2009, Ren2008a}. 

\begin{proposition}\label{prop:fixedsize}
For a ring topology ($L_1$), if $|\alpha| = \pi/N$, the formation converges to a regular polygon with a constant size
\end{proposition}
\begin{proof}
When a symmetric configuration is achieved:
\begin{eqnarray}
	\x_{i+m} - \x_{i} = R(-2\pi m/N)(\x_{i} - \x_{i-m})
\end{eqnarray}
Then
\begin{eqnarray}
\mathbf u_i &=& R(\pi/N)(\x_{i+1} - \x_i) + R(-\pi/N)(\x_{i-1} - \x_i)\nonumber \\
&=& R(\pi/N)(R(-2\pi/N)(\x_{i} - \x_{i-1})) + R(-\pi/N)(\x_{i-1} - \x_i)\nonumber \\
&=& R(-\pi/N)(\x_{i} - \x_{i-1}) + R(-\pi/N)(\x_{i-1} - \x_i) = 0.
\end{eqnarray}
\end{proof}
Similarly, for a more general cyclic topology with state and/or time dependent parameters $k$ and $\alpha$, the condition for $k_m(\bar \x,t)$, $\alpha_m(\bar \x,t)$ to achieve a regular polygon with fixed size is given by:
\begin{eqnarray}
\mathbf u_i &=& \sum_m k_m(\bar \x,t) (R(\alpha_m(\bar\x,t))(R(-2\pi m /N)(\bar\x,t_{i} - \bar\x,t_{[i-m]})) \nonumber \\ && \qquad\qquad\qquad\qquad\qquad\qquad + R(-\alpha_m(\bar\x,t))(\bar\x,t_{[i-m]} - \bar\x,t_i) ) = 0\nonumber \\
&=& \sum_m k_m(\bar\x,t) (R(\alpha_m(\bar\x,t) - 2\pi m /N) - R(-\alpha_m(\bar\x,t))) = 0
\end{eqnarray}

%%%%%%%%%%%%%%%%%%%%%
\subsection{Distributed global convergence to a desired formation size}\label{sec:GlobalSize}

The previous section addresses the problem of converging to a formation under a general cyclic interconnection. However, the subspace to which convergence is defined allows for the size of the formation to be an uncontrolled state of the system.
In general the problem of global convergence to a {splay-state} formation using only neighbor information has been sought after in the literature. As mentioned in the introduction, this article considers the approach to a formation without constraining the relative states to be a specified vector.

One approach that converges to formations without specifying fixed relative vectors in a global frame consists of using structural potential functions \cite{Olfati-saber2002}. When using only relative magnitude information, if the interconnection is not a rigid graph, the global convergence to the desired formation is impossible due to the ambiguity of the possible equilibrium configurations.  
 An extra piece of information, allows the approach discussed in this paper achieving global convergence results, namely, the agreement on an orientation which in the case of spacecraft flight can be achieved by individual star trackers.
In the case of the cyclic pursuit, an approach similar to the one presented in previous work \cite{Ramirez2009}, where the angle of the formation is defined as a dynamic variable that depends on the relative distance to the neighbor seems a reasonable approach however, the stability results are only local.

In this section, using an extension to the approach in the previous section we present a distributed control law for which \emph{global} convergence to the desired size can be guaranteed. Specifically, a sufficient condition in the magnitude of an arbitrary function is determined to guarantee convergence to the desired formation from any initial conditions.

The overall structure of the proof consists of first showing that a
sufficient condition on the bounds of an arbitrary odd function $f(x)$
guarantees convergence to a symmetric formation, i.e. convergence to
the invariant manifold $\M_n$. And then, to show that the trajectories
within that manifold lead to a formation of the desired size.

\begin{theorem}{\bf Global convergence to a regular formation of a desired size} \label{thm:GlobalSize} \\
Consider a set of agents with first order dynamics $\dot \x_i = \uu_i$, interconnected under an undirected cyclic topology with control law:
\begin{eqnarray}\label{eq:uGlobalSize}
u_i = R_{\pi/N}(\x_{i+1} - \x_i) + R^\top_{\pi/N}(\x_{i-1} - \x_{i}) + f(|\x_{i+1} - \x_i| - \rho)(\x_{i+1} - \x_i)
\end{eqnarray}
where $f(z)$ is an arbitrary bounded odd function of $z$, such that $zf(z)>0$ for $z \neq 0$  and $f(0)=0$. The overall dynamics can be written as:
\begin{eqnarray}\label{eq:GlobalSizeSystem}
\dot \x = (-\LL_s + G(\x))\x
\end{eqnarray}
Global convergence to $\M_{n\rho}$, the manifold of regular formations with intervehicle distance $\rho$, $\bar \x \in \M_{n\rho} = \{\bar \x : V_{rn}\x =0  , |\x_i-\x_j |=\rho\ \forall i,j \}$ is guaranteed if:
\begin{eqnarray}\label{eq:GlobCondition}
\lambda_{min}(\bar V_{rn} \LL_s (\bar V_{rn})^\top )  > Nf_{max} \lambda_{max}(\bar V_{rn} A_1 (\bar V_{rn})^\top)
\end{eqnarray}
where $\bar V_{rn}$ is the matrix of orthonormal bases for $V_{rn}$, $\LL_s$ is the symmetric circulant rotational Laplacian, $A_{1}$ is the matrix $A_1 = \begin{bmatrix} -I& I& 0 & \ldots \\ 0& 0& 0 & \ldots \\ \vdots &\vdots &\vdots &\vdots \end{bmatrix}$  and $f(x) < f_{max}$.
\end{theorem}

\begin{proof}
To start, consider the following two lemmas:
\begin{lemma}
The manifold $\M_n$ is a flow-invariant manifold of the dynamics in eq. (\ref{eq:GlobalSizeSystem}) 
\end{lemma}
\begin{proof}
It has been shown that $\M_n$ is an invariant manifold of  $\LL_s$, namely $V_{rn}\LL_s\bar \x = 0$ for $\bar \x$ such that $V_{rn} \bar \x = 0$. For a regular formation $|\x_{k+1}-\x_{k}| = |\x_{i+1}-\x_{i}|$ then $f(|\bar\x_{k+1}-\bar\x_{k}| -\rho)=f( |\bar\x_{i+1}-\bar\x_{i}|-\rho)$, and thus without any loss of generality $G(\bar \x) = (f(| \bar \x_{2}- \bar \x_{1}|) L_1\otimes I_3) \bar \x$.  

Then, $V_{rn}(-\LL_s+ G(\bar \x))\bar \x =  f(|\bar \x_{2}- \bar \x_{1}|) V_{rn}( L_1\otimes I_3) \bar \x = 0$. 
\end{proof}
\begin{lemma}
$\eig(V_{rn}A_{is}\bar V_{rn}^\top) = \eig(\bar V_{rn}A_{ks}\bar V_{rn}^\top)$ for all $i, k \in \{ 1,\ldots,N \}$, where $A_{is} = (A_i+A_i^\top)$.
$A_i = a^{(i)}\otimes I_3$, $a^{(i)}=\{a^{(i)}_{i,j}\} \in \reals^{N\times N}$ is a matrix of zeros except the elements $a^{(i)}_{i,i} = - 1$, $a^{(i)}_{i,i+1} = 1$. $A_1$ was explicitly described above.
\end{lemma}
\begin{proof} 
$\eig(\bar V_{rn}A_{is}\bar V_{rn}^\top) = \eig(\bar V_{rn}A_{ks}\bar V_{rn}^\top)$ if and only if there exists a similarity transformation such that:
\begin{equation}T \bar V_{rn}A_{is}\bar V_{rn}^\top = \bar V_{rn}A_{ks}\bar V_{rn}^\top T\end{equation}
Notice that $T_{ik}A_i = A_kT_{ki}$ where $T_{ik} = T^\top_{ki} = ( I_{n} - L_{k-i})\otimes I_2$, for example, in the case of $i=1, k=2$:
\begin{eqnarray} 
T_{12}A_1 &=& A_2 T^\top_{12}\\
\begin{bmatrix} 0& I& 0 & \ldots \\ 0& 0&I & \ldots \\ \vdots &\vdots &\vdots &\vdots \\ I& 0 & \ldots & 0 \end{bmatrix}\begin{bmatrix} -I& I& 0 & \ldots \\ 0& 0& 0 & \ldots \\ 0& 0& 0 & \ldots \\ \vdots &\vdots &\vdots &\vdots \end{bmatrix}& =& \begin{bmatrix}  0& 0& 0 & \ldots \\ 0 &-I& I& \ldots \\ 0& 0& 0 & \ldots \\ \vdots &\vdots &\vdots &\vdots \end{bmatrix} \begin{bmatrix} 0& 0& 0 & \ldots & I \\ I& 0& 0 & \ldots \\ 0& I& 0 & \ldots \\ \vdots &\vdots &\vdots &\vdots \end{bmatrix}
\end{eqnarray} 
and correspondingly $T_{ik}A_{is} = A_{ks}T_{ki}$.
Then we have that $\bar V_{rn} \bar V_{rn}^\top = I $, therefore by defining $T = \bar V_{rn}T\bar V_{rn}^\top$:
\begin{eqnarray}
T \bar V_{rn}A_{is}\bar V_{rn}^\top = \bar V_{rn}T\bar V_{rn}^\top \bar V_{rn}A_{is}\bar V_{rn}^\top
	 = \bar V_{rn}T A_{is}\bar V_{rn}^\top 
\end{eqnarray}
and
\begin{eqnarray}
\bar V_{rn}A_{ks}\bar V_{rn}^\top T = \bar V_{rn}A_{ks}\bar V_{rn}^\top \bar V_{rn}T\bar V_{rn}^\top 
	 = \bar V_{rn}A_{ks}T\bar V_{rn}^\top;
\end{eqnarray}
showing the desired equivalence.
\end{proof}
Based on the above two lemmas, we can then show that condition (\ref{eq:GlobCondition}) guarantees convergence to the invariant manifold $\M_n$. Specifically, invoking the results in \cite{Wang2005} and \cite{Pham2007} introduced in Section \ref{sec:IntroToApproach} the convergence to the manifold $\M_n:\{\bar \x: \bar V_{rn} \bar \x\}$ is guaranteed if the projected Laplacian of the auxiliary system $\dot \y = (-\LL_s + G(\x) )\y $ is negative definite, i.e.
\begin{eqnarray}\label{eq:ProjLapGlob}
\bar V_{rn}(-\LL_s + G(\x)) \bar V_{rn}^\top < 0
\end{eqnarray}
This can be guaranteed if
\begin{eqnarray}
\lambda_{min}(\bar V_{rn} \LL_s \bar V_{rn}^\top )  &>&  \lambda_{max}(\bar V_{rn} G(\x) \bar V_{rn}^\top)
\end{eqnarray}
Since
\begin{eqnarray}
\lambda_{max}(\bar V_{rn} G(\x) (\bar V_{rn})^\top) &=&  \lambda_{max}(\bar V_{rn} \sum_{i=1}^N f(\x_{i+1}-\x_{i})A_i(\x) (\bar V_{rn})^\top)\\
&<&  N f_{max} \lambda_{max}(\bar V_{rn} A_i(\x) (\bar V_{rn})^\top)
\end{eqnarray}
eq. (\ref{eq:GlobCondition}) guarantees eq. (\ref{eq:ProjLapGlob}) and thus, convergence to $\M_n$.

Now, it is shown that convergence to $\M_n$ implies convergence to $\M_{n\rho}$.  Having that $\LL \bar \x = 0$ then, the dynamics of the distance between any two vehicles in $\M_n$ are given by:
\begin{eqnarray}
\frac{d}{dt}|\x_{i+1} - \x_i| &=& \frac{(\x_{i+1} - \x_{i})^\top}{|\x_{i+1} - \x_i|} (\dot \x_{i+1} - \dot \x_{i})\nonumber \\
&=& \frac{(\x_{i+1} - \x_{i})^\top}{|\x_{i+1} - \x_i|} f(|\x_{i+2} - \x_{i+1}| - \rho)(\x_{i+2} - \x_{i+1}) - f(|\x_{i+1} - \x_{i}|-\rho)(\x_{i+1} - \x_{i})\nonumber\\
&=& \frac{(\x_{i+1} - \x_{i})^\top}{|\x_{i+1} - \x_i|}  ( f(|\x_{i+1} - \x_{i}| - \rho) R_{2\pi/N}(\x_{i+1} - \x_{i}) - f(|\x_{i+1} - \x_{i}| - \rho) (\x_{i+1} - \x_{i})\nonumber\\
&=& \frac{(\x_{i+1} - \x_{i})^\top}{|\x_{i+1} - \x_i|} (f(|\x_{i+1} - \x_{i}| - \rho) (R_{2\pi/N} - I_3) (\x_{i+1} - \x_{i})\nonumber\\
&=& -\cos(\pi/N) f(|\x_{i+1} - \x_{i}| - \rho) |\x_{i+1} - \x_i|
\end{eqnarray}
defining $z=|\x_{i+1} - \x_{i}| - \rho$: 
\begin{eqnarray}
\dot z &=&  -\cos(\pi/N)f(z)(z+\rho).
\end{eqnarray}
A Lyapunov function candidate for this system is $V=\frac{1}{2}z^2$, yielding
\begin{eqnarray}
\dot V &=&  -\cos(\pi/N)zf(z)(z+\rho)
\end{eqnarray}
Since $f(z)$ is an odd function, $zf(z) > 0$ for $z \neq
0$. Furthermore, $z+\rho = |\x_{i+1} - \x_{i}| > 0$, so that $\ \dot V
\leq 0$. Using Lasalle's Invariant Set Theorem~\cite{Slotine_Li}, the
system (\ref{eq:GlobalSizeSystem}) globally converges to the largest
invariant set where $\dot V = 0$, namely $\M_{n\rho}$.
\end{proof}
Note that the global guarantee in this control approach is defined by an upper bound on the arbitrary function $f(z)$. This bound is easily implementable by a saturation function or an arctangent function.

The two results presented in this section generalize results of cyclic
control approach to nonlinear systems. Specifically, an analysis approach that
achieves \emph{global} guarantees for a generalized version of
cyclic pursuit is introduced,  which includes time-varying and state-dependent gains and
coupling matrices as well as more general cyclic
interconnections. Additionally, we introduce a decentralized control
approach with \emph{global} convergence guarantees to a regular
formation of specified size, described by upper bounds on a function that defines the
separation between vehicles.

%%%%%%%%%%%%%%%%%%%%%%%%%%%%%%%%%%%%%%
\subsection{Extension to  second order systems} \label{subsec:secondOrder}

In the derivation of the results we define the systems as first order systems. %In general, for space applications we should consider control approaches for second order systems.
In this section an approach based on a sliding mode control is presented, which shows an straighforward extension of the first order integrators to more complex dynamics. 

Let $\x_i(t)=$ $[x_{i,1}(t),x_{i,2}(t),x_{i,3}(t)]^{\top}\in
\mathbb{R}^3$ be the position at time $t\geq 0$ of the $i$th
agent, $i\in\{1,2,\ldots,n\}$, and let $\x=[\x_1^\top ,\x_2^\top,\ldots,\x_n^\top ]^\top $. 

Consider a linear first order system:
\begin{equation}
\dot \x_i = \sum A_{ij} (\x_{j} - \x_{i})
\end{equation}
shown to converge to manifold $\M_1$.

If the dynamics of each agent are now described by a second order model:%double-integrator model:
\begin{eqnarray}\label{eq:double_int}
\dot \x_i &=&  \vv_i \nonumber \\
\dot \vv_i &=&  f(\x_i, \vv_i) + \uu _i
\end{eqnarray}
consider the feedback control law:
\begin{equation}\label{eq:dynA}
\begin{split}
\mathbf \uu _i = -f(\x_i, \vv_i) + k_d \,\sum A_{ij} (\x_{j} - \x_{i}) &+ \sum A_{ij} (\vv_{j} - \vv_{i}) - k_d\,\vv_i, \quad k_d\in \reals_{>0}.
\end{split}
\end{equation}
A useful form to describe the second order system is by using the sliding variables:
\begin{eqnarray}
\s_i &=& k_d \x_i + \vv_i
\end{eqnarray}
Then, the dynamics of the overall system (\ref{eq:double_int}) with control law (\ref{eq:dynA}) can be writen as:
\begin{eqnarray*}
\ddot \x_i &=& \sum A_{ij} (k_d(\x_{j} - \x_{i}) + (\vv_{j} - \vv_{i})) - k_d \vv_i \\
          &=& \sum A_{ij} ((k_d \x_{j} + \vv_{j}) - (k_d \x_{i}+ \vv_{i}) ) - k_d \vv_i 
%\dot \s_i &=& \sum A_{ij}( \s_{j} - \s_{i} )
\end{eqnarray*}
Then, the system dynamics are described by the equations:
\begin{eqnarray}
\dot \s_i &=& \sum A_{ij}( \s_{j} - \s_{i} ) \label{eq:sCP} \\
\dot \x_i &=& - k_d \x_i + \s_i \label{eq:1ordFilt}
\end{eqnarray}
The first equation describes a first order system for which global convergence to a manifold $\M \subset R^3$ can be analysed following the approach in the previous sections, the second equation is an stable first order filter with input $\s_i$ and output $\x_i$. Then, the trajectories of the agents under control law (\ref{eq:dynA}) are the filtered response of trajectories of the first order system.

Under the control law in eq. (\ref{eq:dynA}), the physical trajectories converge to trajectories which are just the response of the filter $ G_{k_d}(j\omega) = \frac{1}{1 + jk_d\omega}$ to the trajectories of the first order system with initial conditions $\s(0) = k_d \x(0) + \dot \x(0)$.

%This last subsection provides an approach for the implementation of the control laws described in this article for the case of spacecraft formations, where $f_(\x, \vv)$ can be a local description of the gravitational dynamics with respect to some reference frame.

%%%%%%%%%%%%%%%%%%%%%%%%%%%%%%%%%%%%%%%%
\section{Convergence Primitives}\label{sec:FormPrim}
%%%%%%%%%%%%%%%%%%%%%%%%%%%%%%%%%%%%%%%%
An approach to formation control based on combination of primitives is described in this section. This approach is an extension of theorem \ref{thm:PartialContraction}.
Using the idea of primitives, controllers that converge to more complex subspaces can be designed and their global convergence properties verified.

\begin{theorem}\label{thm:multimanifold}
Consider the system:
\begin{eqnarray*}
	\dot \x = \sum_{i}{\f_i(\x)}
	\end{eqnarray*}
such that:
\begin{eqnarray*}
   V_i:\reals^{n} \rightarrow \reals^{p_i}, \quad V_i \f_i(\bar \x) = 0,  \quad V_i \bar \x = 0 , \quad  \forall\ \bar \x \in \M_i
\end{eqnarray*} 
with $rank(V_i)=p_i$.

Then, if either:
\begin{eqnarray}\label{eq:cond1}
i.)\quad\quad\quad\quad    \sum_{i}V{\frac{\partial}{\partial x} \f_i(\x)}V^\top &<& 0  
\end{eqnarray}
%\begin{eqnarray*}
%   \x \in \bigcap_{i}\M_i \quad as \quad t\rightarrow \infty
%\end{eqnarray*}
%Furthermore, if 
%\begin{eqnarray}\label{eq:cond1}
%   V \frac{\partial}{\partial x} \f_i(\x)V^\top &\leq& 0  \quad \forall i
%\end{eqnarray}
%and $V \sum_{i}{\frac{\partial}{\partial x} \f_i(\x)}V^\top$ is full rank or at least for one $i$, $V \frac{\partial}{\partial x} \f_i(\x)V^\top  <  0$, or
or,
\begin{eqnarray}\label{eq:MultiJac}
ii.)\quad\quad\quad\quad  \begin{pmatrix} V_1\frac{d f_1}{d \x}V_1^\top & V_1\frac{d f_2}{d \x}V_2^\top & \ldots & V_1\frac{d f_n}{d \x}V_n^\top \\ V_2\frac{d f_1}{d \x}V_1^\top & V_2\frac{d f_2}{d \x}V_2^\top & \ldots & V_2\frac{d f_n}{d \x}V_n^\top \\ \vdots & \vdots & \vdots & \vdots \\ V_n\frac{d f_1}{d \x}V_1^\top & V_n\frac{d f_2}{d \x}V_2^\top & \ldots & V_n\frac{d f_n}{d \x}V_n^\top \end{pmatrix}<0
\end{eqnarray}
where $span\{V^\top\} = span\{ [V_1^\top V_2^\top ... V_n^\top]\}$. Then:
\begin{eqnarray*}
   \x \in \bigcap_{i}\M_i \quad as \quad t\rightarrow \infty
\end{eqnarray*}
\end{theorem}

\begin{proof} For the first condition consider $\bar \x$, a particular solution such that $ 0 = V_i \bar\x = V_j\bar\x = \ldots=V_n \bar\x$. 
Since $V$ is full row rank, there exists a linear transformation $V=T \bar V$ where $\bar V$ is an orthonormal partition of $\reals^n$ and $TVAV^\top T^\top < 0 \Leftrightarrow VAV^\top < 0$. Then, if condition (\ref{eq:cond1}) holds, the system $\dot \y = V \dot \x$ is contracting with respect to $\y$ and any trajectory of the system $\bf y$ converges to the same solution, namely $\x \rightarrow \bar \x \in \bigcap_{i}\mathcal{N}(V_i) = \bigcap_{i} \M_i$. 

For the second sufficient condition, consider the auxiliary system:
\begin{eqnarray}
 \y = \begin{pmatrix} \y_1 \\ \y_2 \\ \vdots \\ \y_n \end{pmatrix} = \begin{pmatrix} V_1 \\V_2 \\ \vdots \\ V_n \end{pmatrix} \x 
\end{eqnarray}
then:
\begin{eqnarray}
	\dot \y_i = \sum_{k}{\bar V_i \f_k (\bar V_k^{\top} \y_k + U_k^\top U_k\x)}
\end{eqnarray}
where $\bar V_i$ are orthonormal projections of the state such that $\bar V_i\bar \bar V_i^{\top} = I_{p_i}$, and $\bar V_i^{\top} \bar V_i + U_i^\top U_i = I_n$. The Jacobian of the system $\dot \y$ with respect to $\y$ is:
\[\begin{pmatrix} \bar V_1\frac{d f_1}{d \x}\bar V_1^{ \top} & \bar V_1\frac{d f_2}{d \x}\bar V_2^{ \top} & \ldots & \bar V_1\frac{d f_n}{d \x}\bar V_n^{ \top} \\ \bar V_2\frac{d f_1}{d \x}\bar V_1^{ \top} & \bar V_2\frac{d f_2}{d \x}\bar V_2^{ \top} & \ldots & \bar V_2\frac{d f_n}{d \x}\bar V_n^{ \top} \\ \vdots & \vdots & \vdots & \vdots \\ \bar V_n\frac{d f_1}{d \x}\bar V_1^{ \top} & \bar V_n\frac{d f_2}{d \x}\bar V_2^{ \top} & \ldots & \bar V_n\frac{d f_n}{d \x}\bar V_n^{ \top} \end{pmatrix}\]
Since $rank(V_i) = p_i$,  $V_i = T_i \bar V_i$ where $T_i$ is an invertible matrix, it can be writen as:
\[ diag[T_i] \begin{pmatrix} V_1\frac{d f_1}{d \x}V_1^\top & V_1\frac{d f_2}{d \x}V_2^\top & \ldots & V_1\frac{d f_n}{d \x}V_n^\top \\ V_2\frac{d f_1}{d \x}V_1^\top & V_2\frac{d f_2}{d \x}V_2^\top & \ldots & V_2\frac{d f_n}{d \x}V_n^\top \\ \vdots & \vdots & \vdots & \vdots \\ V_n\frac{d f_1}{d \x}V_1^\top & V_n\frac{d f_2}{d \x}V_2^\top & \ldots & V_n\frac{d f_n}{d \x}V_n^\top \end{pmatrix} diag[T_i^\top] \]
and $\diag[T_i]$ is an invertible transformation,  such that $T V A V^\top T^\top<0$ if and only if $V^\top A V^\top < 0$. 
Then, the negative definiteness in eq. (\ref{eq:MultiJac}), proves the global convergence to $y=0$, i.e $ \x \rightarrow \bar \x \in \bigcap_{i}\mathcal{N}(V_i) = \bigcap_{i} \M_i$.
\end{proof}

Notice that additionally, if $V \frac{\partial}{\partial \x} \f_i(\x)V^\top \leq 0$ then $   V \sum_{i}{\frac{\partial}{\partial \x} \f_i(\x)}V^\top$ is at least semidefinite negative. If one of the summation terms is positive definite or if the summation is full rank, it is positive definite.

\begin{corollary}\label{cor:cor1}
Consider a set of $N$ agents with dynamics $\dot \x_i=\uu_i$ grouped in sets $\g_s$ $s \in \{m,n,nm\}$, and a set of control laws $\uu_m=f_{m}(\x),  \uu_n=f_{n}(\x), \uu_{mn}= f_{mn}(\x)$ corresponding to each group where $f_s(\x)$ depends only on elements of set $\g_s$ and has corresponding invariant manifolds $\M_s$ corresponding to the set of transformations $V_s$, for which $V_s f_s = 0$.
Control law $f_m$ interconnect agents in set $S_m \in S$, control law $f_n$ interconnect agents in the set $S_n \in S$, $S_m \cap S_n = \emptyset$, and control law $f_{nm}$ interconnect agents in set $S_n$ and $S_m$.

If $f_s$,  individually converge to their invariant subspaces $\M_s$, i.e. $V_s\frac{df_s}{d \x}V_s^\top < 0$ for all $i$, then, a sufficient condition for the global convergence of the system $\dot \x = \sum_s f_s(\x)$ to the subspace $\M = \bigcap_i \M_i$ is:
\[ \begin{pmatrix} V_k\frac{d f_k}{d \x}V_k^\top & V_l \big(\frac{d f_{kl}}{d \x} + \frac{d f_l}{d \x}^\top \big)V_{kl}^\top \\ V_{kl} \big(\frac{d f_l}{d \x} + \frac{d f_{kl}}{d \x}^\top \big) V_{l}^\top & V_l \frac{d f_l}{d \x} V_l^\top \end{pmatrix} > 0 \]
for all $k,l \in s$. This result can be extended to more than two sets of disjoint groups $n$,$m$ with interconnecting links ${nm}$.
\end{corollary}
\begin{proof}
Since $g_m \cap g_n = \emptyset$, $V_nf_m = 0$, the Jacobian has a block tridiagonal structure:
\[ J = \begin{pmatrix} V_m\frac{d f_m}{d \x}V_m^\top       & V_m\frac{d f_{ml}}{d \x}V_{ml}^\top    &0                                   &  0 & \ldots & 0 \\ 
                       V_{ml}\frac{d f_m}{d \x}V_m^\top    & V_{ml}\frac{d f_{ml}}{d \x}V_{ml}^\top & V_{ml}\frac{d f_{l}}{d \x}V_l^\top &  0 & \ldots & 0 \\
                       0                                   & V_{l}\frac{d f_{ml}}{d \x}V_{ml}^\top   & \ddots                            & \vdots & \vdots \\
                       0 & 0& \ldots & V_{np}\frac{d f_{np}}{d \x}V_n^\top & V_n\frac{d f_n}{d \x}V_n^\top \end{pmatrix} \]
given the block tridiagonal structure the positive definiteness result can be verified by verifying the positive definiteness of lower dimensional matrices following the following proposition.
\end{proof}

\begin{proposition}\label{prop:tridiagonal}
A block tridiagonal matrix with positive block entries $A_{ii}$:
\[ A = \begin{pmatrix} A_{11} & A_{12} & 0 & 0 & \ldots &0 \\ A_{21} & A_{22} & A_{23} & 0 &...&0  \\ 0 & A_{32} & A_{33} & A_{34} &...&0 \\ \vdots & \vdots & \vdots & \vdots & \vdots & \vdots \\ 0& \ldots & 0 & 0 & A_{n,n-1} & A_{n,n} \end{pmatrix}\]
and $A_{11}>0$, $A_{nn}>0$, is positive definite if the submatrices
\begin{equation}\label{cond:subPD}
\begin{pmatrix} A_{ii} & 2A_{i,i+1} \\ 2A_{i+1,i} & A_{i+1,i+1} \end{pmatrix}>0
\end{equation}
\end{proposition}
\begin{proof}
by partitioning the space as $\x = [x_1 \quad x_2\  ... \ x_n]$, the quadratic form is:
\begin{footnotesize}
\begin{eqnarray*}
\x^\top A \x &=& x_1^TA_{11}x_1 + x_1^T(A_{12}+A^\top_{21})x_2 + x_2^TA_{22}x_2 + x_1^T(A_{23}+A^\top_{32})x_2 + x_3^TA_{33}x_3 +\dots \\
	&=& \frac{1}{2}x_1^TA_{11}x_1 + \frac{1}{2}x_1^TAx_1 + x_1^T(A_{12}+A^\top_{21})x_2 + \frac{1}{2}x_2^TA_{22}x_2 + \frac{1}{2}x_2^TA_{22}x_2 +  \ldots + \frac{1}{2}x_n^TAx_n\\
	&=& \frac{1}{2}x_1^TAx_1 + \frac{1}{2} \begin{pmatrix} x_1 \\ x_2 \end{pmatrix}^\top\begin{pmatrix} A_{11} & 2A_{12} \\ 2A_{21} & A_{22} \end{pmatrix}\begin{pmatrix} x_1 &\x_2 \end{pmatrix} + \frac{1}{2}\begin{pmatrix} x_2 \\ x_3 \end{pmatrix}^\top\begin{pmatrix} A_{22} & 2A_{23} \\ 2A_{32} & A_{33} \end{pmatrix}\begin{pmatrix} x_2 &\x_3 \end{pmatrix} + ... + \frac{1}{2}x_n^TAx_n
\end{eqnarray*}
\end{footnotesize}
which is positive if condition in eq. (\ref{cond:subPD}) is met.
\end{proof}

In the next section, we illustrate the application of the results in this section with a series of examples where application of the theorem and the discussed corollaries give insight into the construction of different convergence mechanisms.

%%%%%%%%%%%%%%%%%%%%%%%%%%%%%%%%%%%%%%%%%%%%
\section{Examples}\label{sec:applications}
%%%%%%%%%%%%%%%%%%%%%%%%%%%%%%%%%%%%%%%%%%%%

In a first example we present a useful application of the analysis approach to define a decentralized control algorithm based on a set of primitives whose global convergence properties can be verified from the results of theorem \ref{thm:multimanifold}.

\begin{example}{\bf Global convergence to formation with only relative information}
Consider a group of 8 agents $\x_i,\ i \in \{1,..N\}$, and a control law:
\begin{eqnarray}
	\uu_i = k_1(R_l(t)(\x_{i+1}-\x_{i}) + R_l^\top(t)(\x_{i-1}-\x_{i})) \quad for \quad i \in \g_l
\end{eqnarray}
with agent groups defined as $\g_1=\{1,2,3,4\}, \g_2=\{3,4,5,6\}, \g_3=\{5,6,7,8\}$, and where $R_{\pi/4}(t) \in \reals^{3x3}$ is a rotation matrix that by switching its rotation axis can switch the orientation of the formation. (The switching rate should be slow enough compared to the convergence rate to allow for the formation to convege to a fixed subspace before switching).
Namely, the rotation matrices can be defined as $R_l = T_l(t)R_{\pi/4}T^\top_l(t)$, where $T_1$ is a constant arbitrary direction matrix  (we can assume $T_1=I$ without loss of generality):
\begin{eqnarray*}
T_2 &=& \begin{pmatrix} 1& 0& 0\\ 0& 1& 0\\ 0& 0& -1\end{pmatrix} T_1\\
T_{12} &=& \begin{pmatrix} 0&\cos(\phi)&-\sin(\phi)\\ 0&\sin(\phi)& \cos(\phi) \\ 1&0&0\end{pmatrix} T_1
\end{eqnarray*}
And the subspace constraints are:
\begin{eqnarray*}
V_1 &=& \begin{pmatrix} V_{r4}(I_n\otimes T_1) &\mathbf 0_{6 \times 6}\end{pmatrix}\\
V_2 &=& \begin{pmatrix} \mathbf 0_{6 \times 3}&V_{r4}(I_n\otimes T_{12})& \mathbf 0_{6 \times 3} \end{pmatrix}\\
V_3 &=& \begin{pmatrix} \mathbf 0_{6 \times 6}&V_{r4}(I_n\otimes T_2)  \end{pmatrix}
\end{eqnarray*}
From theorem \ref{thm:multimanifold}, a numerical verification of the global convergence and stability to the specific subspace, namely a cube formation with a face pointing at a given direction defined by $\phi$, is posible.
Figure \ref{fig:Cube}, shows a simulation of the dynamics the formation converging to a cubic formation. At $t=15$, transformaion T whitches and the orientation of the formation changes. The global orientation state of the formation is contralled by a parameter change in only 4 of the vehicles.

\begin{figure}
\centering
\includegraphics[width =0.19\textwidth]{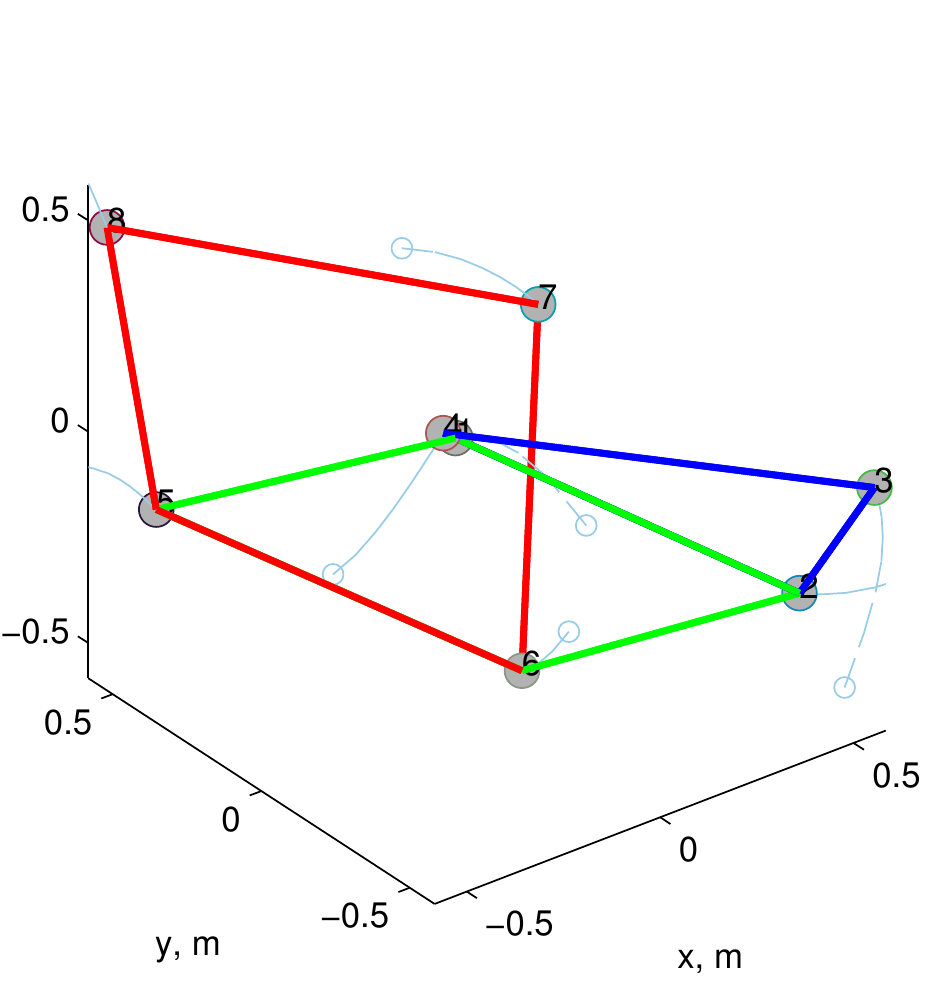}
\includegraphics[width =0.19\textwidth]{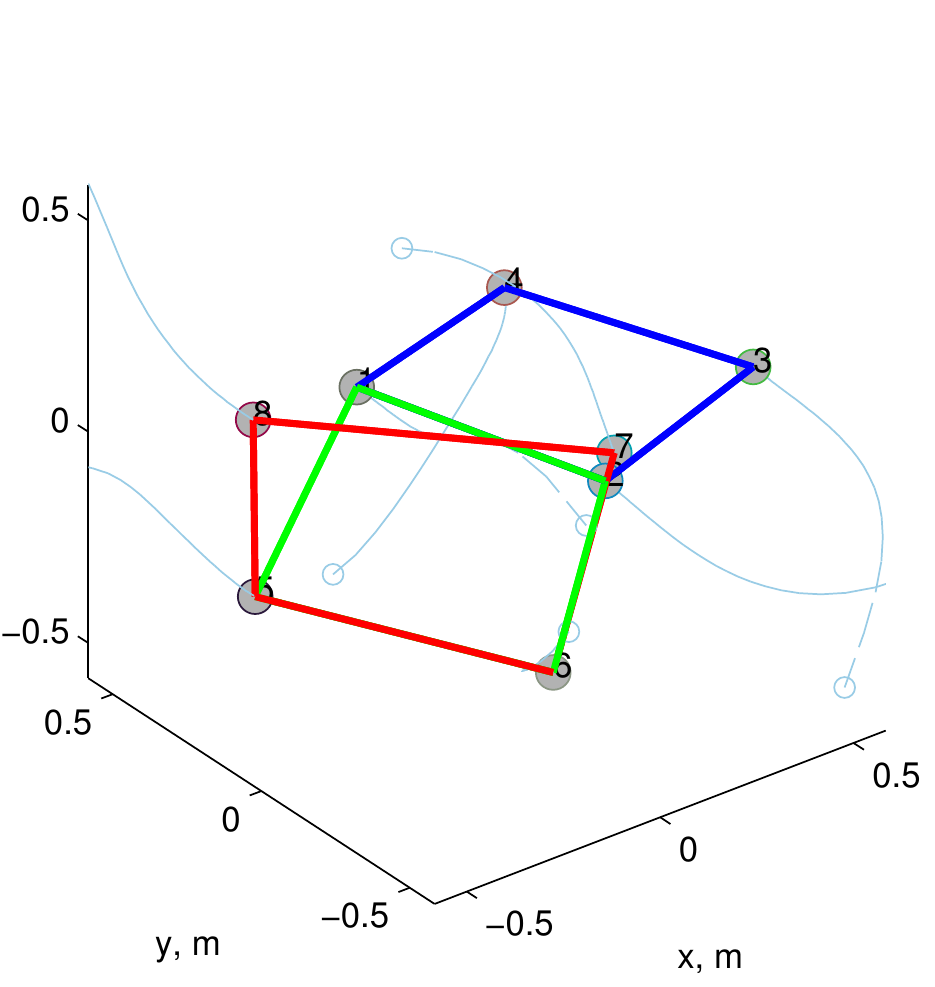}
\includegraphics[width =0.19\textwidth]{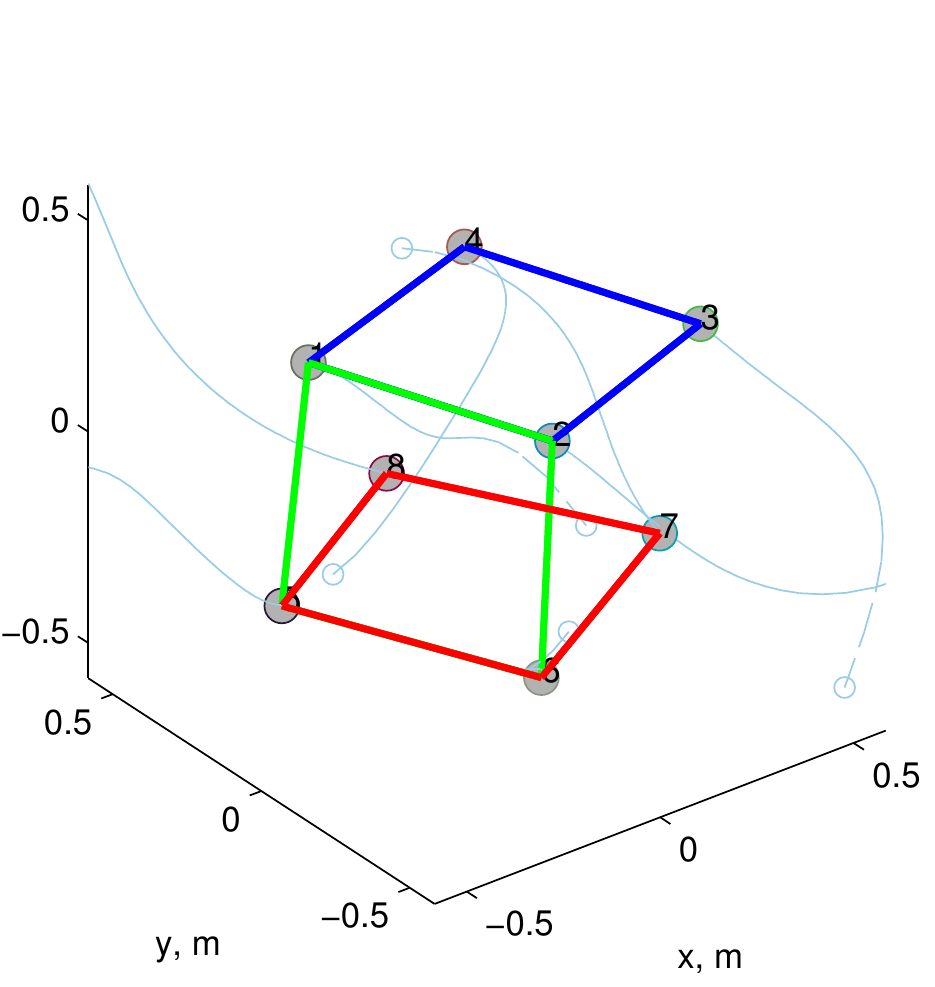}
\includegraphics[width =0.19\textwidth]{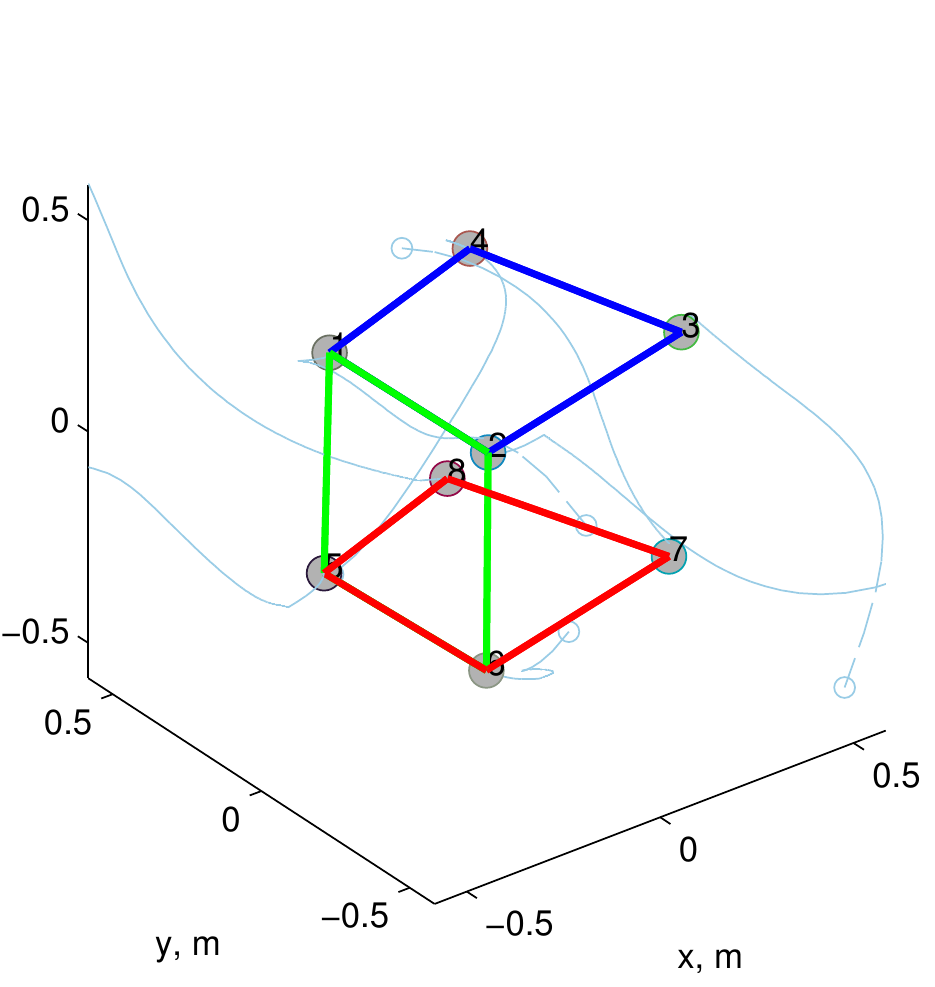}
\includegraphics[width =0.19\textwidth]{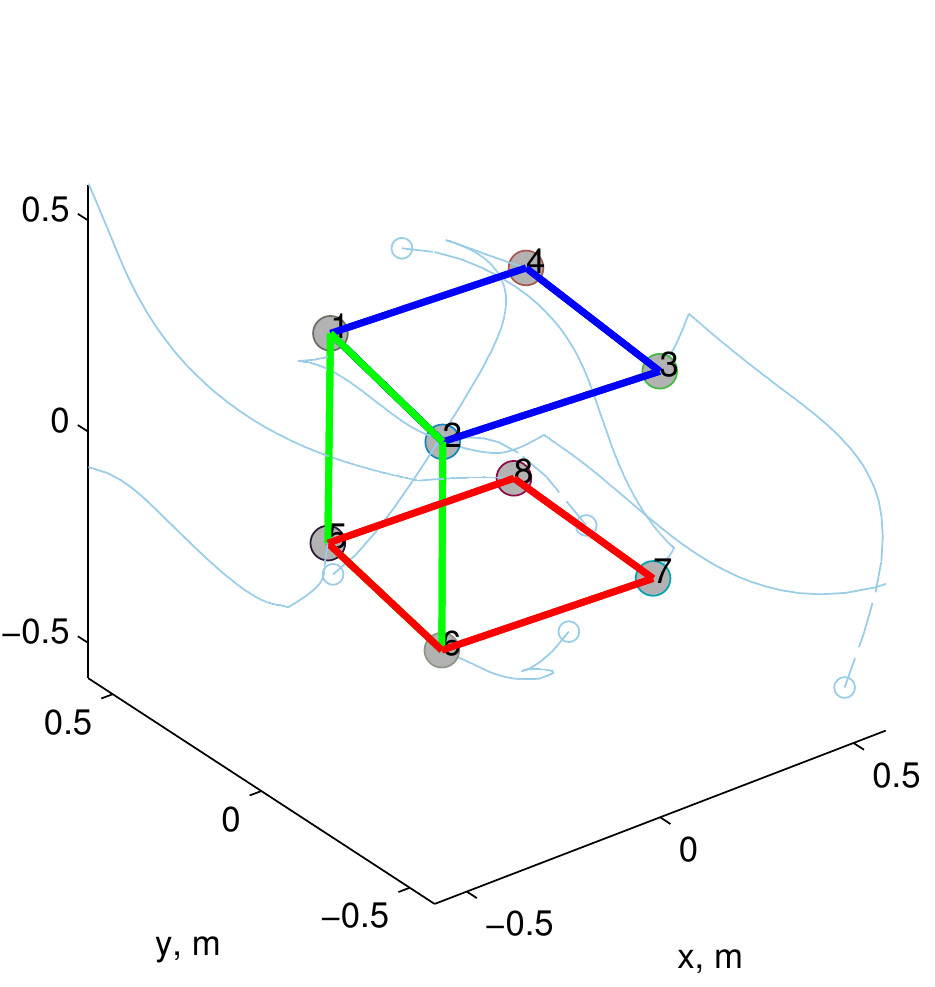}\\
\includegraphics[width =0.8\textwidth]{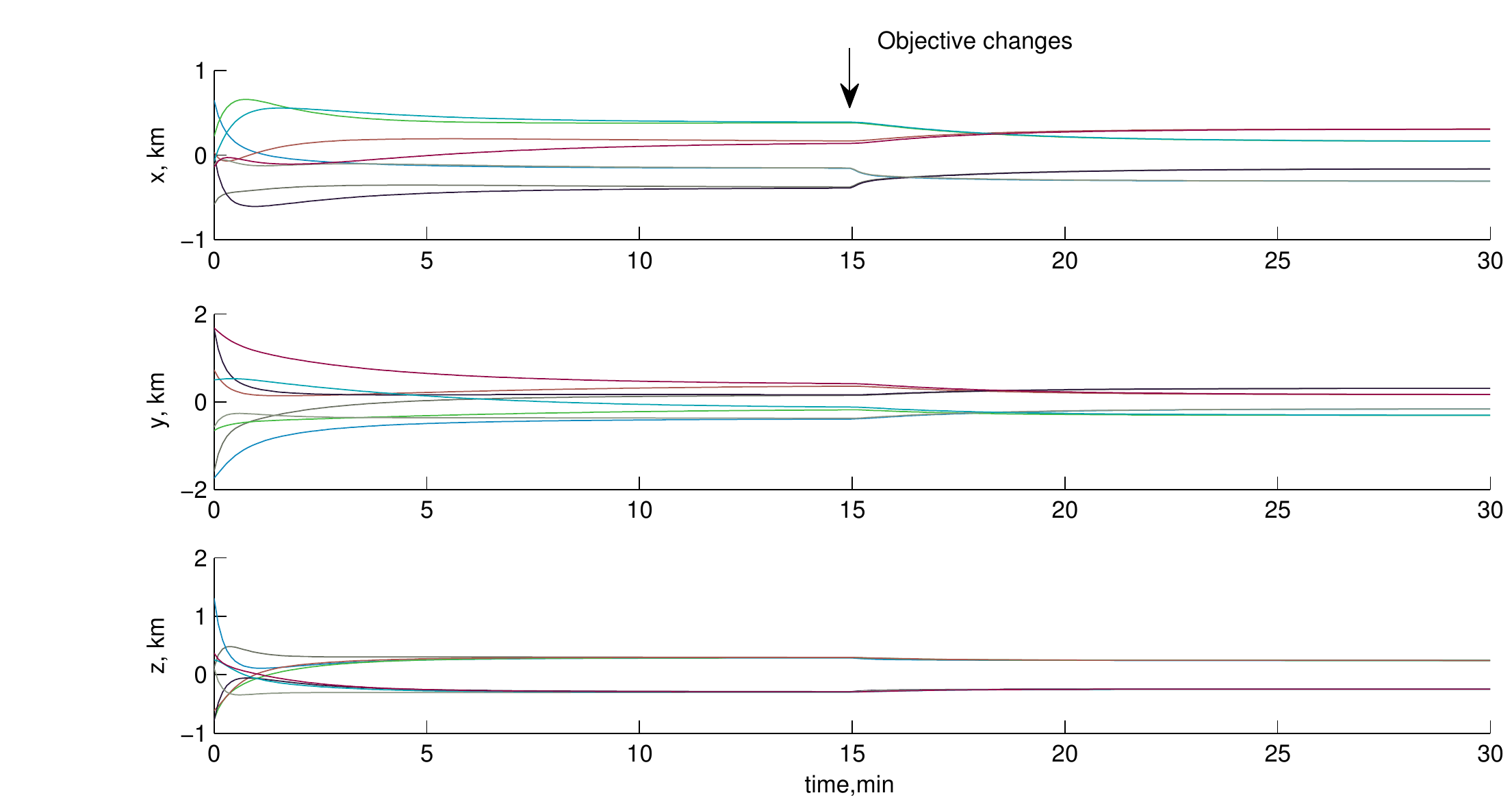}
\caption{Sequence of snapshots agents converge to a cube in 3D and rotate. Trajectory time history below \label{fig:Cube} }
\end{figure}
\end{example}

Another example that illustrates the application to derive convergence properties for general cooperative control problems consists of studying a formation of vehicles surrounding a target:

\begin{example}{\bf A regular formation surrounding an (un)cooperative target}
Consider a system of $n+1$ agents, $n$ of them with dynamics as in eq. (\ref{eq:controllaw}), and a leader agent with state $x_1$ that uses information from all others and/or all other agents use information from it:
\begin{eqnarray*}
	\dot \x_1 &=& f_p(\x_1,\x)\\
	\dot \x_i &=& R(\x_{i+1}-x_{i}) + R^\top(\x_{i-1}-\x_i) + f_c(\x_i-\x_1)
\end{eqnarray*}
$i=\{2,..,n\}$, and some general fucntions $f_c, f_p$ $f_c(0)=0$. Denote $\x = \{\x_2, \x_3,\ldots ,\x_n\}$. 
The overall dynamics can then be written as:
\[	 \begin{pmatrix} \dot \x_1 \\ \x \end{pmatrix} =  \begin{pmatrix} f_p (\x_1,\x) \\ \LL_s\x + f_c(\x,\x_1) \end{pmatrix} \doteq F(\xi)\]
We are interested to determine convergence to the subspace defined by $V\bar \xi = 0$, where: 
\[V = \begin{pmatrix} 0 & V_{rn} \\ nI &-\mathbf 1^T_n\otimes I \end{pmatrix} \]

where $\mathbf 1_n \in \reals^n$ is a vector of ones, $[n I_2\ - \mathbf 1^T_n\otimes I]$ is a projection matrix into a subspace where the target is in the center of the formation and $V_{rn}$ is the projection to the subspace of regular polygons in eq. (\ref{Vrn}).  Considering the following results:
\begin{eqnarray*}
		\LL_s (\mathbf 1_n\otimes I) &=& 0,\\ %(\mathbf 1^T_n\otimes I) \LL_s = 0\\
		V_{rn} \frac{d\f_c}{d\x} [\mathbf 1_n\otimes I] &=& 0 
\end{eqnarray*}
We find that:
\[ \begin{array}{lcr}
	V \frac{dF}{d\xi} V^\top = \begin{pmatrix}  V_{rn}\LL_s V_{rn}^\top + V_{rn}\frac{d f_c}{d \x}V_{rn}^\top & 0 \\ n\frac{d f_p}{d \x}V_{rn}^\top &  n\sum_i(\frac{d f_c}{d \x_i} - \frac{d f_p}{d \x_i}) + n^2 (\frac{d f_p}{d \x_1} - \frac{d f_c}{d \x_1}) \end{pmatrix}
\end{array}\]

Then, if $\frac{d f_c}{d \x} < 0$, a sufficient condition to surround the target is: 
\[\sum_i(\frac{d f_c}{d \x_i} - \frac{d f_p}{d \x_i}) + n(\frac{d f_p}{d \x_1} - \frac{d f_c}{d \x_1})< 0\]
Note that $f_c$, $f_p$ are arbitrary functions and the result gives a sufficient condition to achieve the mission objective in terms only of the gradients of  $f_p$ and $f_c$.  
%\begin{figure}
%	\centering
%		\includegraphics[height=0.25\textheight]{PreyChaseSlow.pdf} \quad
%		\includegraphics[height=0.25\textheight]{PreyChaseFast.pdf}
%\end{figure}
\end{example}

In the next example an example that can be applied to a fragmented aperture system is considered, where individual telescopes are deployed in an arbitrary configuration and the objective consists on achieving convergence to a formation where in its final configuration the vehicles are as close as possible to each while maintaining a minimum separation between them to achieve the recreation of a full aperture composed by many small segments. The problem is related to the two-dimensional sphere packing problem and a solution can be described by a series of concentric hexagonal formations. In this example, a distributed control law based on theorem \ref{thm:multimanifold} is proposed and sufficient conditions for global converge to such configuration with any number of spacecraft are derived.

\begin{example}{\bf Convergence to a packed formation.}
Consider a set of agents with first order dynamics $\x_i=\uu_i$, grouped in M sets $\g_l$ of 6 vehicles, and consider the input:
\begin{eqnarray*}
	\uu_i = R_{\pi/6}(\x_{i+1}-\x_{i}) + R_{\pi/6}^\top(\x_{i-1}-\x_{i}) \f_{m} \quad for \quad i \in \g_m 
\end{eqnarray*}
$\g_m$ being a group of size 6, $m \in \{1,...,M\}$.
This input was shown to make the system of agents $\x_i,\ i \in \g_l$ converge to manifold $\M_6:\{V_{6} \x = 0\}$ of regular hexagonal formations.

Now, consider interconnection control laws between the different sets $\g_l$:
\begin{eqnarray}
	\uu_k = R_{\pi/3}(\x_{k+1}-\x_{k}) + R_{\pi/3}^\top(\x_{k-1}-\x_{k}) = \f_{ml}  \quad for \quad k \in \g_{ml}; 
\end{eqnarray}
$\g_{ml}$ being a group with 3 agents, some agents in $\g_m$ and some in $\g_l$, which has been shown to converge to triangular formations such that $\M_3:\{V_{r3} \x = 0\}$.

The constraints of the desired convergence subspace are $V_m = [V_{r6}\ 0_{4\times 6}]$, $V_l = [0_{4\times 6}\ V_{r6}]$, and the corresponding link constraints $V_{ml} = [ v_1, \ldots, v_N ], v_n \in \reals^{3\times 3}$, $v_{m1} = -I,$, $v_{m2} = I+R_{2\pi/3}$, $v_{l1} = -R_{2\pi/3}$ for some $m1,m2 \in S_m$, $l1 \in \g_l$. 

The Jacobians will be denoted as $\frac{d \f_\theta}{d\x} = A_\theta$. 

Then, from theorem \ref{thm:multimanifold} and corollary \ref{cor:cor1} showing global convergence to a grid defined by a pair of concentric, aligned hexagonal patterns $\g_m$ and $\g_l$ with three-agent link interconnections  between them $\g_{ml}$, requires showing that:
\begin{eqnarray}\label{eq:VerifySP}
 \begin{pmatrix} V_m A_m V_m^\top & V_m \big(A_{ml}+A^\top_{m}\big) V_{ml}^\top \\ V_{ml} \big(A_m+A_{ml}^\top\big) V_{m}^\top & V_{ml} A_{ml} V_{ml}^\top \end{pmatrix} = 
 \begin{pmatrix} V_{r6} \LL_s^{(6)} V_{r6}^\top & V_m (A_{ml}+A^\top_{m}\big) V_{ml}^\top \\ V_{ml} \big(A_m+A_{ml}^\top\big) V_{m}^\top & V_{3} \LL_s^{(3)} V_3^\top \end{pmatrix} &>& 0 \nonumber \\
&&and, \nonumber \\
 \begin{pmatrix} V_{ml} A_{ml} V_{ml}^\top & V_{ml} \big(A_l+A^\top_{ml}\big) V_l^\top \\ V_l \big(A_{ml}+A_l^\top\big) V_{ml}^\top & V_l A_l V_l^\top \end{pmatrix} = 
 \begin{pmatrix} V_{r3} \LL_s^{(3)} V_{r3}^\top & V_{ml} (A_l+A^\top_{ml}\big) V_l^\top \\ V_l \big(A_{ml}+A_l^\top\big) V_{ml}^\top & V_{6} \LL_s^{(6)} V_6^\top \end{pmatrix} &>& 0 \nonumber \\
\end{eqnarray}
where $\LL_s^{(n)} = L \otimes R + L^\top \otimes R^\top$, with $L \in \reals^{n \times n}$.

This result directly verifies \emph{global} convergence for \emph{any}
number of rings with corresponding interconnecting links, and 
overall global convergence of the system is guaranteed by the result in
eq. (\ref{eq:VerifySP}).

Figure \ref{fig:SPgrid} shows snapshots of a a simultion of the dynamics of an example of the convergence for such a controller in a scheme with three hexagonal rings. Agents in $\g_1 = \{1,\ldots,6\}$, $\g_2 =\{7,\ldots,12\}$ $\g_3 =\{13,\ldots,18\}$ converge to hexagonal formations. Formations $\g_{12}=\{1,7,2\}, \g_{12}=\{5,4,10\}, \g_{13}=\{6,17,12\}, \g_{13}=\{3,4,9\}$, establish links that define a manifold $\M = \bigcap \M_i$ corresponding to a regular sphere packing grid. 

\begin{figure}
\centering
\includegraphics[width = 0.4\textwidth]{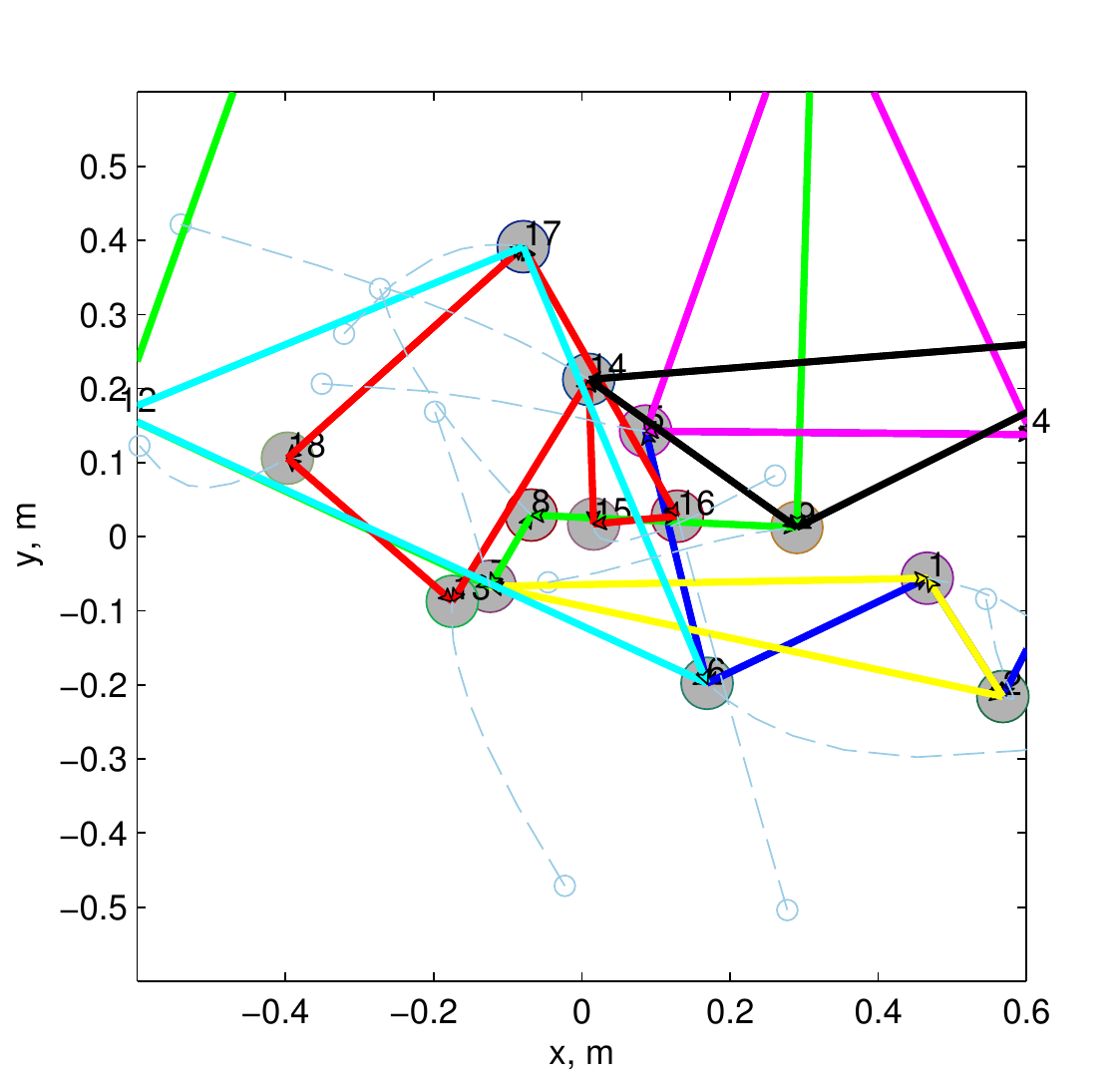}
\includegraphics[width = 0.4\textwidth]{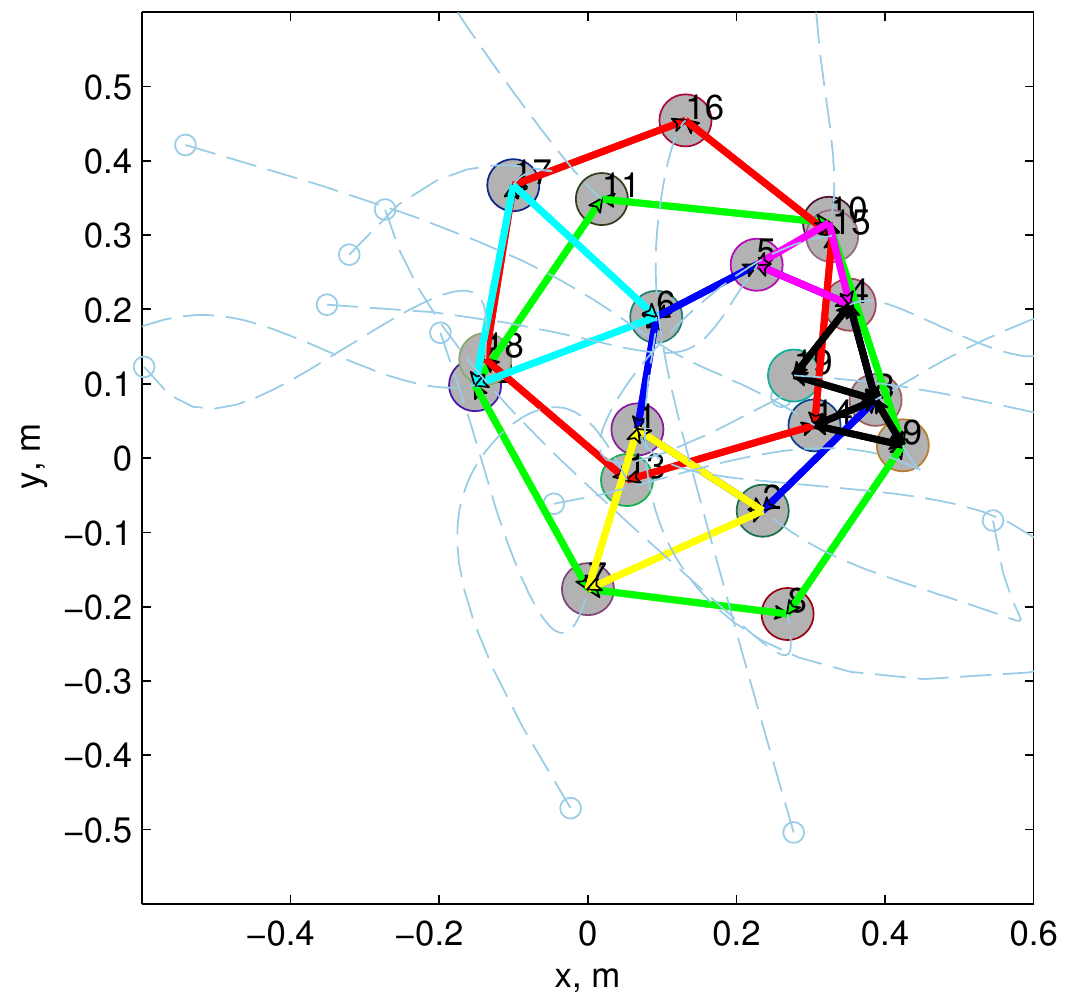}\\
\includegraphics[width = 0.4\textwidth]{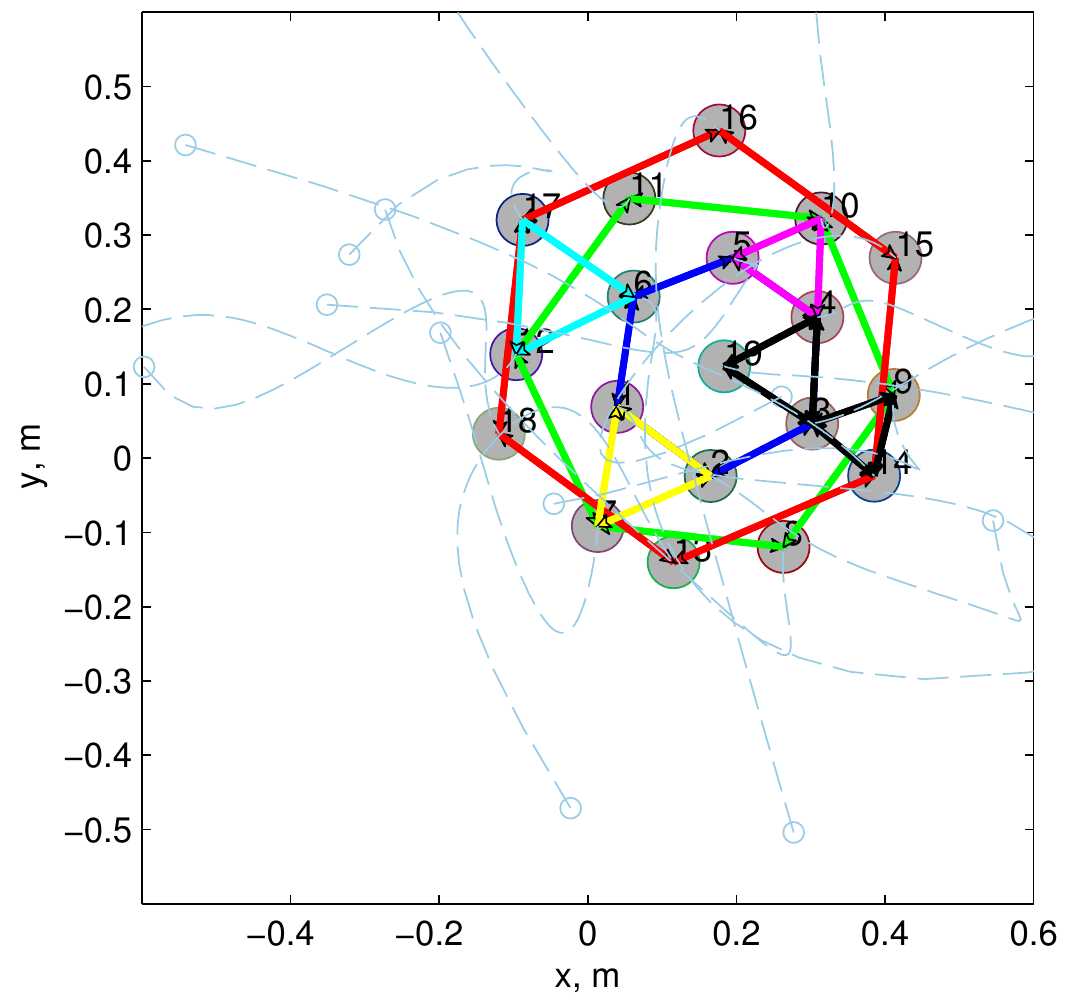}
\includegraphics[width = 0.4\textwidth]{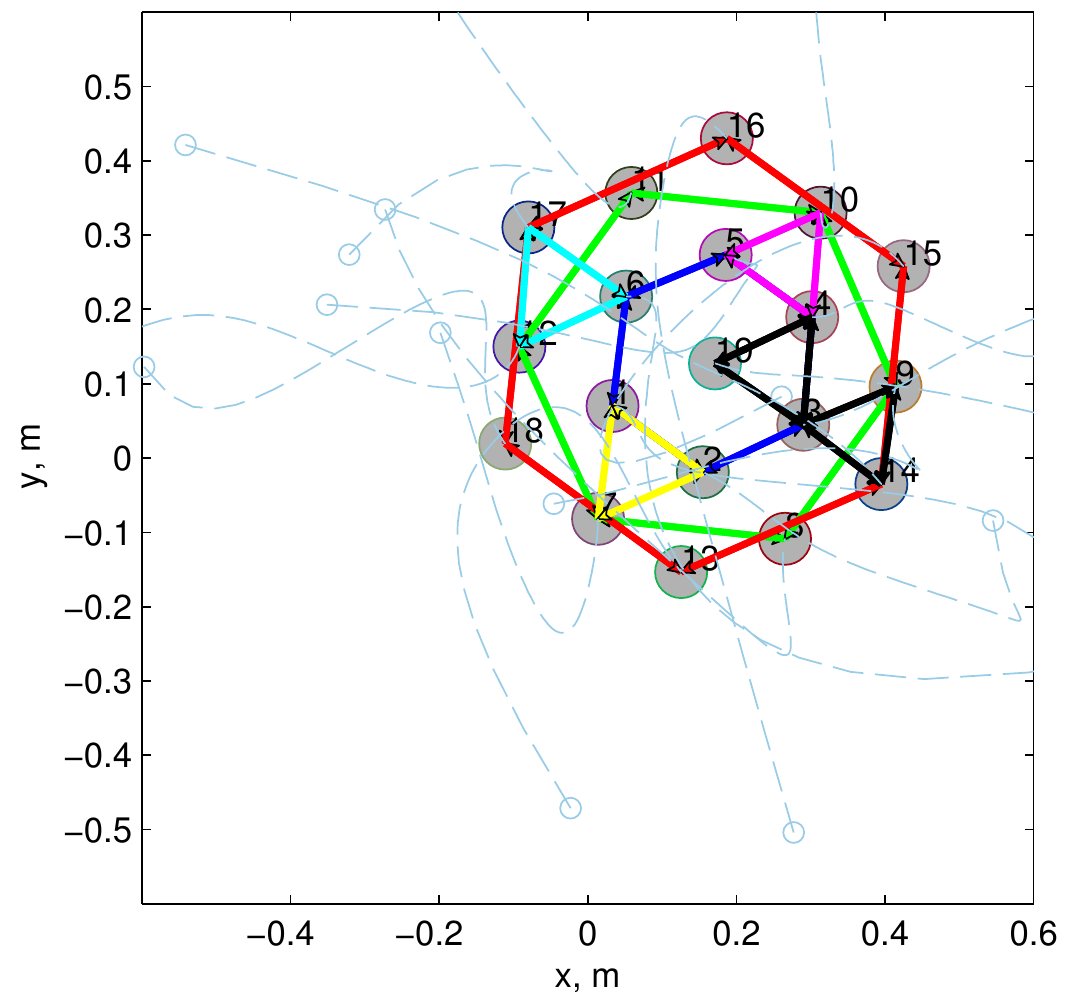}
\caption{Agents converge to a packed grid by imposing some convergence constrains shown by the arrows. \label{fig:SPgrid} T=10s, 90s, 250s, 350s}
\end{figure}

\end{example}

We close this section by presenting an example with a result guaranteing global convergence to a desired size of formation when the size is commanded by one of the spacecraft using only relative information to its neighbor(s).

\begin{example}{\bf Leader based convergence to desired size.}
In this example we consider a system of $n$ agents with a control law (\ref{eq:controllaw}) and a leader that controls its separation to other agents by a control law $f_r(d)$:
\begin{eqnarray}\label{eq:leaderSize}
	u_i &=& R(\pi/N)(x_{i+1}-x_{i}) + R(\pi/N)^\top(x_{i-1}-x_{i})\quad   i={2,3..N} \\
	u_1 &=& R(\alpha)(x_{2}-x_{1}) + R(\alpha)^\top(x_{N}-x_{1}) + f_r(||x_2-x_1||^2-\bar \rho)(x_2-x_1) 
\end{eqnarray}
where $f_r(p)$ is a positive function of $p$ with equilibrium point $0$ such that $f_r(0) = 0$. Then, the system exponentially converges to a symmetric formation with interagent separation $\bar \rho$.

\begin{proof}
The principle of the proof is to show that the dynamics of an auxiliary system $\y=V(\x)$ are contracting, and $\y=0$ is a particular solution of the system. Thus $\y$ tends to zero.

Consider, the overall dynamics of the system of single integrators with control law \ref{eq:leaderSize} :
\begin{eqnarray}\label{eq:SizeCtrl}
\dot \x = -\LL_s \x + f_r(\x)A_1\x
\end{eqnarray}
where $\LL_s$ is the Laplacian defined in previous section which converges to regular formations and $A_1 = \begin{pmatrix} -I& I& 0 & \ldots \\ 0& 0& 0 & \ldots \\ \vdots &\vdots &\vdots &\vdots \end{pmatrix}$.
Consider the auxiliary variables
\begin{eqnarray}
\begin{pmatrix} y_1 \\ y_2 \end{pmatrix} = \begin{pmatrix} V \x \\ f_r(\x) \end{pmatrix}
\end{eqnarray}
where $V^\top V + U^\top U = I$, $\LL_s U^\top U \x = 0$, which define an nonlinear invariant manifold $\M_{n\rho}: \{\x\ |\ V \x=0, f_r(\x) = 0 \}$, since $\dot \x = 0$ for $\x \in \M_r$. Their dynamics are given by:
\begin{eqnarray}\label{eq:auxSys}
\begin{pmatrix} \dot y_1 \\ \dot y_2 \end{pmatrix} &=& \begin{pmatrix} V \dot \x \\ (\nabla_{\x} f_r) \dot \x \end{pmatrix} 
%=  \begin{pmatrix} V \\ (\nabla_{\x} f_r) \end{pmatrix} \dot \x = \begin{pmatrix} V \\ (\nabla_{\x} f_r) \end{pmatrix} \LL_s \x +  \begin{pmatrix} V \\ (\nabla_{\x} f_r) \end{pmatrix} y_2 A_1\x \\
= \begin{pmatrix} V(-\LL_s \x + y_2 A_1\x)  \\ (\nabla_{\x} f_r) (-\LL_s \x + y_2 A_1\x) \end{pmatrix} \nonumber \\
%&=& \begin{pmatrix} V \LL_s V^\top y_1 + y_2 VA_1\x  \\ (\nabla_{\x} f_r) \LL_s V^\top y_1 + (\nabla_{\x} f_r) y_2 A_1\x \end{pmatrix} \nonumber \\
\begin{pmatrix} \dot y_1 \\ \dot y_2 \end{pmatrix} &=& \begin{pmatrix} -V \LL_s V^\top & VA_1\x  \\ -(\nabla_{\x} f_r) \LL_s V^\top & (\nabla_{\x} f_r)A_1\x \end{pmatrix} \begin{pmatrix} y_1 \\ y_2 \end{pmatrix} = \mathbf f(\x,\y)
\end{eqnarray}
The auxiliary system (\ref{eq:auxSys}) is contracting and all trajectories will converge to the same trajectory if:
\begin{eqnarray*}
\mathbf F = \Theta \frac{\partial \mathbf f}{\partial \y} \Theta^{-1} &<& 0
\end{eqnarray*}
Specifically $\y=0$ is a solution, then, any trajectory of system (\ref{eq:auxSys}) converges to $\y=0$, which means that any solution of (\ref{eq:SizeCtrl}) converges to $\x \in \M_{n\rho}$.

$V$ and $V_{rn}$ are related through an invertible transfomation $V=TV_{rn}$, and consider an invertible transformation that commutes with T, i.e $\Theta = \begin{pmatrix}I &0 \\ 0& \theta \end{pmatrix}$, where $\theta > 0 \in R$.  Then,
 \[T \Theta  \frac{d \mathbf f}{d \y} \Theta^{-1} T^\top = \Theta T  \frac{d \mathbf f}{d \y} T^\top \Theta^{-1} < 0 \Leftrightarrow \mathbf F < 0\].

Then, a sufficient condition for convergence of system (\ref{eq:SizeCtrl}) to a regular formation with characteristic size $\bar \rho$ is:
\begin{eqnarray}
J_y &=& \begin{pmatrix} -V_{rn} \LL_s V_{rn}^\top & \theta^{-1}V_{rn}A_1\x  \\ -\theta(\nabla_{\x} f_r) \LL_s V_{rn}^\top & (\nabla_{\x} f_r)A_1\x \end{pmatrix} < 0
\end{eqnarray}
We have shown that $V_{rn} \LL_s V_{rn}^\top < 0$, and we also have that:
\begin{eqnarray*}
%A_1\x = (x_2-x_1)\begin{pmatrix} I \\ 0 \\ \vdots \\ 0 \end{pmatrix}\\
V_{rn}A_1\x = \begin{pmatrix} -I \\ 0 \\ \vdots \\ 0 \end{pmatrix} (x_2 - x_1) = D_1(x_2 - x_1)\\ 
\nabla_{\x}f_r = 2(x_2 - x_1)^\top\frac{d f_r(\rho)}{d\rho}\left[I\ I\ 0\ \ldots 0 \right] = 2\frac{d f_r(\rho)}{d\rho} (x_2 - x_1)^\top D_2\\
%(\nabla_{\x} f_r)A_1\x = 2(x_2-x_1)^2\frac{d f_r(\rho)}{d\rho}
\end{eqnarray*}
%\[ \begin{pmatrix}a  \end{pmatrix} \]

Then, the projected Jacobian $J_y$ is negative definite iff:
\begin{eqnarray}
(\nabla_{\x} f_r)A_1\x = -2(x_2-x_1)^2\frac{d f_r(\rho)}{d\rho}  &<& 0 \nonumber \\ 
\Rightarrow \frac{d f_r(\rho)}{d\rho}>0	
\end{eqnarray}
and (from a Schur factorization):
\begin{small}
\begin{eqnarray*}
(\nabla_{\x} f_r)A_1\x + \frac{1}{2}(\theta^{-1} V_{rn}A_1\x - \theta((\nabla_{\x} f_r) \LL_s V_{rn}^\top))(V_{rn}\LL_s V_{rn}^\top)^{-1}\frac{1}{2}(\theta^{-1}V_{rn}A_1\x - \theta((\nabla_{\x} f_r) \LL_s V_{rn}^\top))^\top <0 \nonumber \\
\Rightarrow \frac{1}{4}\left[ (x_2-x_1)^\top\left( D_1 - 2\theta^{2}\frac{df_r}{d\rho}D \right)(V_{rn}\LL_s V_{rn}^\top)^{-1}\left( D_1 - 2\theta^{2} \frac{df_r}{d\rho}D \right)^T (x_2-x_1)  \right] \qquad
\\
\qquad\qquad\qquad\qquad\qquad< -2\theta^2 \frac{d f_r}{d\rho}(x_2-x_1)^2 \nonumber\\
\Rightarrow \frac{1}{4}\left[ \left( D_1 - 2\theta^{2}\frac{df_r}{d\rho}D \right )( V_{rn}\LL_s V_{rn}^\top)^{-1}\left( D_1 - 2 \theta^2 \frac{df_r}{d\rho}D \right)^T  \right] < -2\theta^2 \frac{d f_r}{d\rho} I
\end{eqnarray*}
\end{small}
$D = D_2\LL_s V_{rn}^\top$, verifying that the negative definiteness of $J_y$ does not depend on the value of $\x$. If for a particular value of $\x$, a constant transformation $\Theta$ is found which verifies the negative definiteness it holds uniformly. 

Then, the satifaction of the two above conditions is sufficient to guarantee convergence in terms of $\frac{d\f}{d\rho}$. For example, for $\frac{d\f}{d\x}=1$, $N=5$, $\Theta = \begin{pmatrix}I &0 \\ 0& 0.5 \end{pmatrix}$, verifies \emph{global} convergence to any desired size of formation.
\end{proof}
\end{example}

\section{Concluding remarks}\label{sec:Conclusion}
A new approach to formation control has been presented by using the tools of contraction theory. The advantage of this approach include global convergence to formations defined by constraints of the desired state and not trajectory points. Further research would consider linear or affine transformation for convergence to not only regular formations, and the use of the convergence primitive tools to define more complex convergence states. Additionally, extending the results to a time varying manifold, seems a feasible  straightforward extension. %Regarding the validation of the algorithms, preliminary results on a testbed are succesful in maintaining a 3 vehicle formation. 
We will consider tests of the algorithms in the SPHERES testbed in the microgravity environment of the ISS. 

%\section*{Acknowledgements}
%The authors would like to thank the SPHERES team at the Space Systems Laboratory at MIT, especially  Dr. Alvar Sanz-Otero and Dr. David Miller, for facilitating the implementation on the testbed and for their support in obtaining the experimental results.

%{\bf Leaders, inhibition}

%{\bf Reference Randy Beard 2001}

%\section*{References}
%\bibliographystyle{aiaa}
%\bibliography{../../Thesis/BibTeX/RefAll}

%%%%%%%%%%%%%%%%%%%%%%%%%%%%%%%%%%%%%%%
%Appendix
%%%%%%%%%%%%%%%%%%%%%%%%%%%%%%%%%%%%%%%

\appendix
\section{Eigenvalues of the projected Laplacian}\label{app:A}
The derivation of the eigenvalues of $\RR_{\eta} \LL^{(N)} \LL_m(\alpha) (\LL^{(N)})^\top \RR_{\eta}^\top$ ( denoted here $\RR \LL \LL_{sm} \LL^\top \RR^\top$ for brevity of notation) uses the properties of block-circulant matrices, specifically, the fact that they all belong to a commutative algebra. We have that:
\begin{eqnarray*}
\eig[\RR \LL \LL_{sm} \LL^\top \RR^\top] 
  &=& \eig[\RR \RR^\top] \eig[\LL^{(N)}] \eig[(\LL^{(N)})^\top] \eig[\LL(\alpha) + \LL^\top(\alpha)]\\
  &=& \eig[(I_n\otimes R(\pi/N))(I_n\otimes R^\top(\pi/N))]\\
  && \times \eig[(L_1\otimes I_3)(L_1^\top \otimes I_3)]\eig[\LL(\alpha) + \LL(\alpha)^\top]\\
  &=& \eig[(L_1\otimes I_3)(L_1^\top \otimes I_3)] \eig[\LL(\alpha) + \LL(\alpha)^\top]
\end{eqnarray*}
Then,
\begin{eqnarray}
\text{eig}[\RR \LL \LL_{sm} \LL^\top \RR^\top]
 &=& \lambda_{ik} \quad i \in \{1,..,N\}, k\in \{-1,0,1\} \nonumber \\
\lambda_{ik} &=& (1-e^{-j2im\pi/N})(1-e^{j2im\pi/N})((e^{jk\alpha}-e^{jk\alpha - j2im\pi/N}) + (e^{jk\alpha} - e^{-jk\alpha + j2i\pi/N})) \nonumber \\
&=&  ( 2 - e^{-j2i\pi/N} - e^{j2i\pi/N} ) 2(\cos(k\alpha) - \cos(k\alpha - 2im\pi/N)) \nonumber \\
&=&  4( 1 - \cos(2i\pi/N))(\cos(k\alpha) - \cos(k\alpha - 2im\pi/N)) \label{eq:eigrLLLR}
\end{eqnarray}
Then, for $i=1$: \ \ \ \ \ \ \ \ \ \ $ (1 - \cos(2i\pi/N)) = 0
\ \ \ \ \ \Rightarrow \ \ \ \ \ \lambda_{1k} = 0$

for $i=N$: \ \ \ \ \ \ \ \ \ \ $ 
(\cos(k\alpha) - \cos(k\alpha-2mi\pi/N)) = 0 \ \ \ \ \ \Rightarrow \ \ \ \ \ \lambda_{Nk} = 0 $

otherwise for $1<i<N$, $|\alpha| \leq 2\pi/N$ : 
\begin{eqnarray*}
(\cos(k\alpha) - \cos(k\alpha-2im\pi/N)) > 0 \\
(1 - \cos(2i\pi/N)) > 0
\end{eqnarray*}
\begin{proposition}\label{prop:EigenLap}
\begin{eqnarray}
V_{rn} \LL_{sm} V_{rn}^\top = W_n \RR_{\eta} \LL \LL_{sm}(\alpha) \LL^\top \RR_{\eta}^\top W_n^\top > 0 \quad,
\end{eqnarray}
\end{proposition}
\begin{proof}
$\RR \LL \LL_{sm} \LL^\top \RR^\top$ is a symmetric matrix, then, its nullity is the algebraic multiplicity of the $0$ eigenvalue. 
Thus, $\dim\{\mathcal{N}(\RR \LL \LL_{sm} \LL^\top \RR^\top)\} = 6$. 
On the other hand, the dimension of $W_n(\RR \LL \LL_{sm} \LL^\top \RR^\top)W_n^\top = 3N-6$, and thus it is full rank.
Since $(\RR_{\eta} \LL \LL_{sm}(\alpha) \LL^\top \RR_{\eta}^\top)$ is positive semidefinite, then $V_{rn} \LL_{sm} V_{rn}^\top$ is at least positive semidefinite, but since it is full rank, the proposition is proven.
\end{proof}
\end{document}